\documentclass [12pt]{article}

\usepackage[utf8]{inputenc}

\usepackage{times, amsmath, amssymb, amsfonts, amssymb, bm, amsthm, dsfont}
\usepackage{amsbsy}
\usepackage[pdftex]{graphicx}
\usepackage{enumitem}
\usepackage{bm}
\usepackage{istgame}
\usepackage{tikz}
\usepackage{mathrsfs}
\usepackage{natbib}
\usepackage[all]{nowidow}

\usepackage{setspace}
\usepackage{hyperref}
\hypersetup{hidelinks}
\hypersetup{
    colorlinks=true,
    citecolor=blue,
    linkcolor=blue,
    urlcolor=blue,
    pdfborder={0 0 0}
}

\usepackage{enumitem}
\usepackage[protrusion=true,expansion=true,nopatch=footnote]{microtype}

\newcommand{\De}{\Delta}

\newcommand{\Ga}{\Gamma}

\newcommand{\si}{\sigma}

\newcommand{\ov}{\overline}
\newcommand{\un}{\underline}
\newcommand{\ta}{\theta}
\newcommand{\Ta}{\Theta}

\newcommand{\R}{\mathbb{R}}

\newcommand{\envheader}{\smallskip\par\medskip\par\nopagebreak}
\newcommand{\envfooter}{\smallskip\par\medskip\par}

\newcounter{theoremc}
\newenvironment{theorem}{\envheader\setlength{\parskip}{0ex}\refstepcounter{theoremc}\noindent\textbf{Theorem~\thetheoremc{}:}\begin{em}}{\end{em}\envfooter}
\newenvironment{theoremC[1]}{\envheader\setlength{\parskip}{0ex}\refstepcounter{theoremc}\noindent\textbf{Theorem~\thetheoremc{} (#1):}\begin{em}}{\end{em}\envfooter}

\newenvironment{theorem*}{\envheader\setlength{\parskip}{0ex}\refstepcounter{theoremc}\noindent\textbf{Theorem:}\begin{em}}{\end{em}\envfooter}

\newcounter{lemmac}
\newenvironment{lemma}{\envheader\setlength{\parskip}{0ex}\refstepcounter{lemmac}\noindent\textbf{Lemma~\thelemmac{}:}\begin{em}}{\end{em}\envfooter}

\newcounter{corollaryc}
\newenvironment{corollary}{\envheader\setlength{\parskip}{0ex}\refstepcounter{corollaryc}\noindent\textbf{Corollary~\thecorollaryc{}:}\begin{em}}{\end{em}\envfooter}

\newcounter{definitionc}
\newenvironment{definition}[1][]{\envheader\setlength{\parskip}{0ex}\refstepcounter{definitionc}\noindent\textbf{Definition~\thedefinitionc{}\if\relax\detokenize{#1}\relax\else\ (#1)\fi:}\begin{em}}{\end{em}\envfooter}
\newenvironment{definitionC[1]}{\envheader\setlength{\parskip}{0ex}\refstepcounter{definitionc}\noindent\textbf{Definition (#1):}\begin{em}}{\end{em}\envfooter}

\newcommand{\one}{\mathbf{1}}

\usetikzlibrary{arrows.meta, decorations.pathreplacing}

\newcounter{propositionc}
\newenvironment{proposition}{\envheader\setlength{\parskip}{0ex}\refstepcounter{propositionc}\noindent\textbf{Proposition~\thepropositionc{}:}\begin{em}}{\end{em}\envfooter}

\newcounter{remarkc}

\newcounter{examplec}

\newcounter{conventionc}

\begin{document}
\setlength{\emergencystretch}{1.2em}
\hypersetup{pdftitle={Recognizing Dominance through Local Improvement}, pdfauthor={Linchen Zhang}}

\newcommand{\websitedomain}{linchenzhang.github.io}
\newcommand{\websiteurl}{https://\websitedomain}
\DeclareRobustCommand{\websitelink}{\href{\websiteurl}{\mbox{\websitedomain}}}

\title{Recognizing Dominance through Local Improvement\thanks{I am deeply indebted to Songzi Du for his invaluable guidance on this project. I am immensely grateful to Aram Grigoryan, Joel Sobel, Joel Watson, and Emanuel Vespa for their useful suggestions for this paper. I also thank Isabel Trevino, Denis Shishkin, Vince Crawford, Suraj Malladi, and participants at the UCSD TBE seminar and workshop for valuable discussions. Claude was used for text refinement and proof verification in this project. All errors are my own.}\par\vspace{0.5em}
\vspace{0.3em}
{\normalsize \href{\websiteurl}{Click here for the latest version}\par}}

\author{Linchen Zhang\thanks{Department of Economics, UC San Diego. Email: \text{liz032@ucsd.edu}. Website: \websitelink.}}
\date{October 29, 2026\thanks{This is a preliminary version. It may be incorporated into a new version with additional experimental results from the experiment proposed in Appendix~\ref{sec:accounts}, which I plan to run in the near future.}}

\maketitle

\begin{abstract}
\noindent Strategy-proofness is the traditional standard for predictable behavior, yet participants often fail to play dominant strategies. Most existing simplicity concepts attribute these failures to the cost of contingent reasoning, which they lower by comparing the worst case of the dominant strategy with the best case of an alternative. I argue that there is another approach, one that makes more use of the structure of the strategy set, which I formalize as a neighbor relation. Comparing neighboring strategies can reduce the need for contingent reasoning, and, even with a single contingency, it can simplify the search through the strategy set. I introduce a new simplicity criterion: a mechanism is local improvement solvable under a neighbor relation if repeatedly moving to an improving neighbor always reaches a dominant strategy. Strategy-proofness is the criterion under the coarsest relation, and many classical mechanisms satisfy it under finer ones. In extensive-form mechanisms without duplicate dominant strategies, the criterion under one-situation neighbors is strict action-proofness, which is distinct from obvious strategy-proofness. Among local improvement solvable mechanisms that implement the same outcomes, the criterion predicts more dominant-strategy play where the local comparison is salient.

\medskip
\noindent \textit{JEL codes:} C72, C91, D44, D47, D82, D91.
\end{abstract}

\newpage

\section{Introduction}
\label{sec:intro}

A mechanism can be relied on to deliver the outcomes it was designed for only if participants behave as predicted. The traditional standard for predictable behavior is strategy-proofness. In a strategy-proof mechanism, every participant has a dominant strategy, one that is optimal whatever the others do. The prediction then does not depend on her beliefs about the others, and it should be relatively simple for the participant to figure out what her dominant strategy is. Yet participants often fail to play their dominant strategies. \citet{kagel87} find that subjects overbid in the sealed-bid second-price auction, and \citet{reesjones18} and \citet{hassidim} find that applicants misreport in high-stakes deferred acceptance clearinghouses. The rate of dominant-strategy play also differs widely across strategy-proof mechanisms, as \citet{cason06} and \citet{chensonmez06} show. Some strategy-proof mechanisms therefore seem to be simpler than others, and strategy-proofness by itself does not say which.

One reason a dominant strategy can be hard to recognize is that recognizing it may require heavy contingent reasoning.\footnote{\citet{niederlevespa23} review the experimental evidence on failures of contingent reasoning.} Suppose you value an object at \$3.50 and bid for it in a sealed-bid second-price auction. Bids are whole dollars from \$0 to \$7, and a tie for the highest bid means no sale. In this setting, your dominant bid is \$4, the lowest bid above your value. To compare \$4 directly with a distant alternative such as \$7, you must reason contingency by contingency. The two bids give you different payoffs only when the highest rival bid is \$4, \$5, or \$6. In each of these cases, \$7 wins at a price above your value, while \$4 does not win. To conclude that \$4 is better than \$7, you must therefore work through three cases, which participants can find burdensome.

In response to failures of contingent reasoning, the recent literature on simple mechanisms looks for mechanisms that require less contingent reasoning. The most influential concept in this literature is the obvious strategy-proofness of \citet{li17}. A strategy is obviously dominant if, at every situation it reaches, its worst possible continuation outcome is at least as good as the best possible outcome of any deviation at that situation. Here a situation is a point at which the player must choose an action. A mechanism is obviously strategy-proof, or OSP, when every type of every player has an obviously dominant strategy. A participant can then recognize dominance by this worst-best comparison, without separating the contingencies. The sealed-bid second-price auction is not OSP, since the best that \$7 can bring, winning at a price of \$0, beats the worst that \$4 can bring, winning nothing. Later concepts, reviewed in Section~\ref{sec:related}, strengthen or weaken obvious strategy-proofness, or rank strategy-proof mechanisms by the contingent reasoning that remains under worst-best comparisons, as \citet{nagelsaitto} do. I refer to obvious strategy-proofness and the later concepts that compare a worst case with a best case as the worst-best comparison concepts.

Worst-best comparisons may not, however, be the only simplifying route that participants can use to recognize a dominant strategy. In this paper, I study a second way to reduce contingent reasoning, the comparison of neighboring strategies. The idea is that two similar strategies, which I formally define as neighbors, lead to the same outcome, or call for the same reasoning, in many contingencies. When a participant compares two neighbors, the contingencies in which they lead to the same outcome cancel out, and less contingent reasoning is needed. I argue that the comparison of neighbors defines a dimension of simplicity that the worst-best comparison concepts do not measure.

For example, return to the sealed-bid second-price auction and let the neighbors of a bid be the bids one dollar above and below it. The bids \$4 and \$3 are neighbors and give you different payoffs only when the highest rival bid is exactly \$3. Then \$4 wins at a price of \$3, worth \$0.50 to you, and \$3 does not win. So \$4 improves \$3, in the sense that \$4 is at least as good whatever the rivals bid and strictly better in at least one case. Similarly, \$4 improves \$5. The two bids give different payoffs only when the highest rival bid is \$4, and then \$5 wins at a price above your value, which gives you a negative payoff. By the same argument, each bid other than \$4 is improved by its neighbor closer to \$4. Moving from any bid to an improving neighbor, as long as one exists, therefore ends at \$4, the dominant bid.

The argument with neighboring bids uses an aspect of the mechanism that the worst-best comparison concepts do not fully exploit: the structure of the strategy set. The strategy set has an order, here the order of bids, that determines which strategies are neighbors. As the auction example shows, the order of strategies can be a further source of simplicity in recognizing the dominant strategy. To formalize the argument, I add one primitive to the standard model, a neighbor relation on each strategy set. The neighbor relation specifies which pairs of strategies are neighbors, and so which pairs are compared locally. A relation of this kind is implicit in the one-shot deviation principle of game theory. The principle, presented for instance in \citet{osborne94}, checks a strategy only against the strategies that differ from it at one situation. A neighbor relation allows other notions of a single change, such as one grid step of a bid or an exchange of two adjacent entries in a ranking.

A strategy improves another if it is at least as good in every contingency and strictly better in some. The local improvement algorithm moves to an improving neighbor as long as one exists. I then propose the key simplicity criterion of this paper, local improvement solvability. A mechanism is local improvement solvable under a neighbor relation if, for every type of every player, every run of the algorithm ends at a dominant strategy, whatever the starting strategy. Local improvement solvability implies strategy-proofness and becomes stronger as the neighbor relation becomes finer, that is, as each strategy has fewer neighbors. Under the complete neighbor relation, in which every two strategies are neighbors, local improvement solvability is exactly strategy-proofness, which is therefore the coarsest member of the resulting family of conditions.

I then show that many classical strategy-proof mechanisms satisfy the criterion under their natural neighbor relations, the relations that the presentations of their strategy sets define. The sealed-bid second-price auction and the multi-unit Vickrey auction are local improvement solvable under one-step neighbors, which change a bid, or one marginal bid, by one grid step. In matching mechanisms in which participants submit rankings, two rankings are one-exchange neighbors if they differ by an exchange of two adjacent entries. On the full domain of rankings, deterministic matching mechanisms of this kind in which truthful reporting is dominant are local improvement solvable under one-exchange neighbors, provided the environment is rich. Richness, which seems a mild condition, requires every exchange of two adjacent entries to change the participant's outcome in some contingency.\footnote{Richness fails, for example, when the participant knows that some adjacent pair on her list cannot affect her outcome.} Serial dictatorship, top trading cycles, and deferred acceptance are such mechanisms. In these classical mechanisms, the comparison between two neighboring strategies involves only one case. They may therefore contain a hidden, shorter path to the dominant strategy, and Section~\ref{sec:dynamic} discusses how a mechanism can guide participants along it.

Local improvement solvability extends to extensive-form mechanisms and games. At a situation $h$, the extensive-form structure restricts the comparison to the strategies and contingencies that reach $h$. The reduced normal-form mechanism does not show this restriction, which the worst-best comparison concepts use. Local improvement solvability can also use a finer relation. Under one-situation neighbors, two strategies are neighbors if they differ in the action at exactly one situation. A step of the local improvement algorithm then changes one action and holds the rest of the strategy fixed. Such a step is a one-shot deviation, so one-situation neighbors are the relation implicit in the one-shot deviation principle.

Action-proofness, a condition I introduce, concerns exactly these one-shot deviations.\footnote{Action-proofness has precedents in the agent normal form of \citet{selten75} and in the weak dominance between actions of \citet{catoninixue}.} It requires that, at each situation the recommended strategy reaches, no one-shot deviation from the recommended action pay more, whatever the player does at her other situations and whatever the others do. Theorem~\ref{prop:hamming} shows that local improvement solvability under one-situation neighbors is equivalent to strict action-proofness, absent duplicate dominant strategies. The theorem can thus be read as a one-shot deviation principle for dominance. In a strictly action-proof mechanism, every strategy with no improving one-shot deviation is dominant. Under action-proofness the participant need not plan her own later actions, but she must still compare actions contingency by contingency, which obvious strategy-proofness lets her avoid. Neither condition implies the other. The strong obvious strategy-proofness of \citet{pyciatroyan23}, or SOSP, implies both action-proofness and obvious strategy-proofness, although the two together do not imply SOSP.

The worst-best comparison concepts and local improvement solvability can disagree, as three example games in Section~\ref{sec:games} show. In the first example, the pick-a-number game, two players each pick a number from $\{0, 1, \dots, K\}$, and each receives her own number plus $\lambda$ times the other player's number, where $\lambda > K$. Raising one's number by one raises one's payoff by exactly one, whatever the other player picks, so the game is local improvement solvable under one-step neighbors and seems easy. Yet picking $K$ is not obviously dominant, since its worst payoff, $K$, falls short of the best payoff $a + \lambda K$ of any smaller number $a$. The measure of \citet{nagelsaitto} classifies the game as maximally complex, and no strategy-proof extensive-form game that implements the same outcomes is OSP. The worst-best comparison concepts seem to mislabel this easy game as hard, while local comparisons explain its simplicity. Chess, when White has a winning strategy, and Sudoku with a unique solution show the opposite error. Both games seem difficult, yet chess is OSP for White and Sudoku is SOSP, so both are simple by obvious strategy-proofness. On the other hand, neither game is local improvement solvable under its natural neighbor relation, one-situation neighbors in chess and one-cell neighbors in Sudoku. In each game, some strategy with two mistakes has no improving neighbor, although that strategy is not dominant. Local improvement solvability therefore accounts for part of what the worst-best comparison concepts miss.

Sudoku also shows that local improvement solvability concerns more than contingent reasoning. A Sudoku puzzle has a single contingency, since there is no other player and no chance move, so finding the solution requires no contingent reasoning at all. Its difficulty seems to lie instead in the search for the solution among the many fillings of the grid. With a single contingency, local improvement solvability requires every strategy without an improving neighbor to be a best strategy, as in an optimization problem in which every local maximum is global. The neighbor relation thus captures a source of simplicity beyond contingent reasoning. It concerns the search through the strategy set, and the worst-best comparison concepts do not measure it. The example in Section~\ref{sec:lis} that separates local improvement solvability from strategy-proofness makes the same point with three strategies and a single contingency.

Local improvement solvability also yields a behavioral prediction. Comparisons between neighbors help a participant only if she makes them. A mechanism makes the local comparison salient when it asks her for a choice that contains a comparison between two neighboring strategies, or when it resolves a submitted strategy step by step. A request for her whole strategy at once does neither. I therefore predict that, of two strategy-proof mechanisms that implement the same outcomes and are local improvement solvable for the same reason, the one that makes the local comparison salient produces more dominant-strategy play.

The ascending clock auction is OSP and produces the same outcomes as the sealed-bid second-price auction. \citet{li17} finds more dominant-strategy play in the ascending clock auction and attributes the difference to obvious strategy-proofness. Yet the ascending clock auction also asks at each price whether to exit now or stay one step longer, a comparison of two neighboring bids. Obvious strategy-proofness and the salience of the local comparison are thus two separate channels through which the ascending clock auction could help the participant recognize the dominant strategy. Li's comparison between the two auctions changes both channels at once.

A later experiment separates the two channels. \citet{bsk22} add to the design of \citet{li17} a passive clock treatment, in which subjects submit a sealed bid and then watch an ascending clock, driven by that bid, determine the outcome. The passive clock has the same extensive-form mechanism as the sealed-bid second-price auction and is not OSP, so obvious strategy-proofness on its own predicts no difference between the two. Yet it makes the local comparison salient. It removes about half of the difference in deviations from truthful bidding between the sealed-bid second-price auction and the ascending clock auction in which bidders see their rivals exit. Dynamic bidding, the feature that makes the ascending clock auction OSP, accounts for a further twelve percent, which is not statistically significant. Showing bidders when their rivals exit accounts for almost forty percent, which I attribute to learning from the environment. The results suggest that, in the ascending clock auction, the local comparison does more than obvious strategy-proofness to help participants recognize the dominant strategy.\footnote{Section~\ref{sec:limits} discusses another explanation of the passive clock's gain, learning from feedback.}

Table~\ref{tab:evidence} in Appendix~\ref{app:evidence} collects further experimental evidence on the effectiveness of obvious strategy-proofness and of the local comparison. The evidence seems favorable to both channels on balance. Each appears to help participants recognize the dominant strategy, substantially in some environments and marginally in others, although in some settings the gain does not carry over to outcomes.

The local improvement criterion leads to two design recommendations. The first is to choose, among mechanisms that implement the same outcomes, one that is local improvement solvable under its natural neighbor relation. The second is to make the local comparison salient. Local improvement solvability resembles a short solution to a difficult mathematics problem, which a student may miss on her own but find with the right hint. A presentation that makes the local comparison salient is a hint of this kind. If the salience channel is real, the hint is also cheap, since a presentation that only makes the local comparison salient leaves the incentive properties of the mechanism unchanged.

The rest of the paper is organized as follows. Section~\ref{sec:related} reviews the related literature, including the closest antecedent, \citet{catoninixue}. Sections~\ref{sec:model} to~\ref{sec:extensive} define the criterion and apply it to classical and extensive-form mechanisms. Section~\ref{sec:games} examines the three games, Section~\ref{sec:dynamic} states the prediction and reviews the evidence, and Section~\ref{sec:design} gives the design recommendations and concludes. Proofs omitted from the main text are in Appendix~\ref{app:proofs}, Appendix~\ref{app:evidence} collects the experimental evidence, and Appendix~\ref{app:discussion} contains further discussion, including a proposed design.

\section{Related literature}
\label{sec:related}

The literature on the worst-best comparison concepts begins with the obvious strategy-proofness of \citet{li17}. The concepts that followed vary Li's comparison along two dimensions, and Table~\ref{tab:four} in Section~\ref{sec:taxonomy} makes both explicit. The first dimension is how the contingencies are compared. Dominance compares the recommended strategy with a deviation one contingency at a time. Obvious dominance compares the worst case of the recommended strategy against the best case of the deviation. \citet{zhanglevin17} give a decision-theoretic foundation for obvious dominance, and the concepts of \citet{zhanglevin21} and \citet{arribillagamassoneme24} lie between the two. They compare the worst case against the best case within each element of a partition of the contingencies, which \citet{arribillagamassoneme24} derive from a partition of the players. \citet{chewwang22} weaken obvious dominance using the sure-thing principle. The second dimension is how much of her own future play the participant is assumed to plan. \citet{pyciatroyan23} define strong obvious strategy-proofness, under which she plans none of it. They also define a family of concepts indexed by her planning horizon, the set of future decisions she plans for. In this family, obvious strategy-proofness is the case in which she plans for all of them. \citet{nagelsaitto} turn the binary classification into a complete ranking of strategy-proof mechanisms by the contingent reasoning they require.

Two further concepts concern what a participant notices or needs to believe rather than what she must compare. The obvious manipulations of \citet{troyanmorrill20} are the deviations a participant can see to be profitable when she compares worst cases with worst cases and best cases with best cases. In the strategically simple mechanisms of \citet{borgersli19}, she can determine an optimal choice from first-order beliefs about the others' preferences and first-order certainty of their rationality, without higher-order beliefs.\footnote{\citet{li24} reviews three formalizations of simplicity.} None of these concepts uses the structure of the strategy set beyond what an extensive-form mechanism records.

This paper adds the structure of the strategy set as an input and replaces the comparison of the recommended strategy with each alternative by a chain of comparisons between neighbors. The two kinds of argument are complementary. Section~\ref{sec:extensive} shows that strategy-proofness, OSP, and SOSP are conditions on a single step of the local improvement algorithm in extensive-form mechanisms. Local improvement solvability instead requires every run of the algorithm to end at a dominant strategy.

\citet{catoninixue} is the closest antecedent. They compare two actions at one situation rather than two strategies. Given a partition of the contingencies that reach the situation, one action dominates another if every continuation after the second is matched by some continuation after the first. In each element of the partition, the worst payoff of the match must be at least the best payoff of the continuation it matches. Their weak dominance between actions uses the singleton partition. They show that a strategy is dominant if and only if each action it prescribes at a situation it reaches is weakly dominant. The trivial partition gives the corresponding result for obvious dominance.

Their weak dominance between actions is closely related to action dominance, which Section~\ref{sec:extensive} introduces and Theorem~\ref{prop:hamming} relates to local improvement solvability under one-situation neighbors. Both compare two actions at one situation, one contingency at a time. Their condition lets the continuation after the dominating action depend on the continuation it is matched with. Action dominance requires every continuation after the recommended action to do at least as well, in every contingency, as every continuation after the other action. \citet{catoninixue} do not study this stronger condition. Section~\ref{sec:taxonomy} shows that this requirement is exactly what separates action-proofness from strategy-proofness.

The local dominance of \citet{catoninixue} derives the partition from the game. Where the dominating action ends the player's play, its worst payoff must be at least the best payoff of the other action. Elsewhere, ``continuing in the same way'' after the two actions, which repeats after the dominating action the later actions taken after the other one, must give the same outcome. A mechanism is locally strategy-proof if each type has a locally dominant action at every situation she reaches by choosing such actions, and local strategy-proofness implies strategy-proofness. Local dominance is close to the comparison between neighbors. In the English auction with alternating bids, their argument compares leaving now with bidding and leaving at the next turn. Section~\ref{sec:clock} identifies the same comparison, between adjacent exit prices, in the ascending clock auction. Like the natural neighbor relation of this paper, local dominance depends on how the game is presented, since ``continuing in the same way'' is defined by stages and action labels. Appendix~\ref{app:other} relates the conditions of \citet{catoninixue} to those of Table~\ref{tab:four}.

This paper differs from theirs in three respects. First, the comparison is between two strategies that are neighbors under a relation on the strategy set, and such relations are defined for static mechanisms as well. Local strategy-proofness coincides with obvious strategy-proofness in static mechanisms, as they show, so the sealed-bid second-price auction, which is local improvement solvable under one-step neighbors, is not locally strategy-proof. Second, a local improvement is usually judged one contingency at a time, while local dominance compares a worst case with a best case where the dominating action ends the player's play. Third, the object of study is whether repeated local comparisons reach the dominant strategy from every starting strategy. The resulting family of conditions is ordered by the fineness of the neighbor relation, with strategy-proofness as its coarsest member.

The local incentive compatibility literature studies when local incentive constraints imply global ones. \citet{carroll12} studies the question for cardinal and ordinal preferences and with transfers. On ordinal domains the question is also studied in \citet{sato13}, \citet{cho16}, and \citet{kumaretal21}. With transfers it is studied in \citet{archerkleinberg14}, \citet{mishrapramanikroy16}, and \citet{kumarroy24}, and for random assignment in \citet{mennleseuken21}. These papers ask whether an analyst who knows the outcome function can verify strategy-proofness from the local constraints alone. This paper asks whether a participant who compares neighboring strategies reaches the dominant strategy.\footnote{Proposition~\ref{prop:matching}(b) shows that in a deterministic ranking mechanism on the full strict domain that is strategy-proof under truthful reporting, an adjacent exchange changes the outcome only between the two exchanged objects. For deterministic mechanisms, this result is one direction of the characterization of strategy-proofness by swap monotonicity and upper and lower invariance in \citet{mennleseuken21}. It also restates the mistake monotonicity of \citet{cho16} for this environment. Part (b) is the direction from global to local constraints.}

\citet{glazerrubinstein96} show that a static game solvable by iterated elimination of dominated strategies can be presented as an extensive-form game with the same normal-form game. In the extensive-form game, backward induction reaches the same solution through the same comparisons. They argue that the presentation guides the player to the solution. Section~\ref{sec:dynamic} makes the corresponding point for strategy-proof mechanisms, arguing that a dynamic mechanism or a passive presentation guides the participant by making the local comparison salient. Table~\ref{tab:evidence} in Appendix~\ref{app:evidence} collects the experiments that hold the outcomes fixed and vary the mechanism or its presentation.\footnote{Among these experiments, \citet{ght23} also formalize descriptions of a mechanism as directed graphs that compute its outcome. \citet{kagellevin01} report related experiments on multi-unit auctions, and \citet{hakimovkubler21} survey the matching experiments.}

Two further lines of literature concern the errors that remain. \citet{bdk23} define the primitive complexity of a menu as the number of operations needed to find its best element. \citet{oprea20} and \citet{banovetzoprea23} find that the number of states in a rule predicts its cost and use. \citet{enkeshubatt23} and \citet{oprea24} relate errors in choices between lotteries to their complexity. Primitive complexity is the kind of difficulty the Sudoku game of Section~\ref{sec:games} has, and the worst-best comparison concepts classify it as absent. \citet{velezbrown} ask which deviations from the dominant strategy are consistent with empirical equilibrium, a refinement of Nash equilibrium.\footnote{\citet{velezbrown} find that deviations that change the outcome are consistent with empirical equilibrium only when the rule is bossy and information is not interior.} Their criterion asks which deviations survive as limits of noisy behavior in which an action is played more often than another only if it pays more. Local improvement solvability concerns which deviations survive the participant's own comparison of neighboring strategies.

\section{The criterion}
\label{sec:model}

\subsection{Basic setup}
\label{sec:setup}

A mechanism $\Ga$ specifies for each player $i$ a finite strategy set $S_{i}$ and an outcome function $f:S \to \Omega$ that maps strategy profiles $s = (s_{0}, s_{1}, \dots, s_{n}) \in S = S_{0} \times S_{1} \times \dots \times S_{n}$ to outcomes, where $s_{0} \in S_{0}$ is nature's move. Player $i$ has a type $\ta_{i}$ in a set $\Ta_{i}$ and a utility $u_{i}(s_{i}, s_{-i}; \ta_{i})$ over the outcome of each profile. Throughout, the type is carried explicitly as an argument of the utility, and $s_{-i}$ ranges over the contingencies, the profiles of the others' strategies and nature's moves.

The new primitive is a neighbor relation on each strategy set.

\begin{definition}[Neighbor relation]
\label{def:nbhd}
A neighbor relation on $S_{i}$ assigns to each strategy $s_{i}$ a set $N_{i}(s_{i}) \subseteq S_{i} \setminus \{s_{i}\}$ of neighbors, with $s_{i}' \in N_{i}(s_{i})$ if and only if $s_{i} \in N_{i}(s_{i}')$.
\end{definition}

\begin{definition}[Finer and complete neighbor relations]
\label{def:finer}
A neighbor relation $N_{i}$ is finer than $N_{i}'$ if $N_{i}(s_{i}) \subseteq N_{i}'(s_{i})$ for every $s_{i}$. The complete neighbor relation is $N_{i}(s_{i}) = S_{i} \setminus \{s_{i}\}$.
\end{definition}

Three relations on bids illustrate the definitions. Let a strategy be a bid on the grid $S_{i} = \{0, 1, \dots, K\}$. Under one-step neighbors, the neighbors of $b$ are $b - 1$ and $b + 1$. Under two-step neighbors, they are the bids at most two steps from $b$, namely $b - 2$, $b - 1$, $b + 1$, and $b + 2$. Under the complete neighbor relation, they are all other bids. In each case only bids on the grid count. One-step neighbors are finer than two-step neighbors, and both are finer than the complete neighbor relation, as every neighbor relation is.

The choice of neighbor relation matters, since a mechanism can be local improvement solvable under one relation and not under another. Section~\ref{sec:dynamic} argues that a dynamic mechanism can determine which neighbors a participant considers.

\begin{definition}[Dominance and strategy-proofness]
\label{def:dom}
A strategy $s_{i}^{*}$ is dominant for $\ta_{i}$ if
\begin{equation}
u_{i}(s_{i}^{*}, s_{-i}; \ta_{i}) \;\ge\; u_{i}(s_{i}, s_{-i}; \ta_{i}) \qquad \text{for all } s_{i} \in S_{i} \text{ and all } s_{-i} \in S_{-i}.
\label{eq:SPdef}
\end{equation}
The mechanism $\Ga$ is strategy-proof (SP) if every type of every player has a dominant strategy.
\end{definition}

Two strategies are payoff-equivalent for a type if they give her the same utility in every contingency. Two dominant strategies must be payoff-equivalent, since each is at least as good as the other in every contingency.

\begin{definition}[Recommendation]
\label{def:rec}
A recommendation is a profile of maps $\si_{i} : \Ta_{i} \to S_{i}$. The mechanism $\Ga$ is strategy-proof under $\si$ if each $\si_{i}(\ta_{i})$ is dominant for $\ta_{i}$.
\end{definition}

In a direct mechanism, where $S_{i} = \Ta_{i}$, the recommendation is truthful reporting, $\si_{i}(\ta_{i}) = \ta_{i}$. In the ascending clock auction of Section~\ref{sec:clock} it is to exit at the first price at or above the value. In the dynamic serial dictatorship of Section~\ref{sec:dynamicmatching} it is to take the top available object. The results below concern whether comparisons between neighboring strategies lead to $\si_{i}(\ta_{i})$. They do not use that $\si_{i}(\ta_{i})$ is a reported type.

\begin{definition}[Improvement]
\label{def:improves}
A strategy $s_{i}'$ improves $s_{i}$ for $\ta_{i}$ if $u_{i}(s_{i}', s_{-i}; \ta_{i}) \ge u_{i}(s_{i}, s_{-i}; \ta_{i})$ for every $s_{-i} \in S_{-i}$, with strict inequality for some $s_{-i} \in S_{-i}$.
\end{definition}

No strategy improves itself, and no strategy improves a dominant strategy. If $s_{i}''$ improves $s_{i}'$ and $s_{i}'$ improves $s_{i}$, then $s_{i}''$ improves $s_{i}$. 

\subsection{Local improvement solvability}
\label{sec:lis}

\begin{definition}[Local improvement]
A strategy $s_{i}'$ locally improves $s_{i}$ for $\ta_{i}$ if $s_{i}' \in N_{i}(s_{i})$ and $s_{i}'$ improves $s_{i}$. Write $s_{i} \to s_{i}'$.
\end{definition}

In words, a local improvement is a neighbor that improves the current strategy, one at least as good in every contingency and strictly better in some. Compare the local improvement with (\ref{eq:SPdef}). The local comparison is the same pointwise comparison that strategy-proofness makes, restricted to a pair of neighbors, so it does not require beliefs about the others. In the applications of Section~\ref{sec:classical}, the local comparison reduces to one case, the contingencies in which the two neighbors give the participant different outcomes.

\begin{definition}[Local improvement algorithm]
The local improvement algorithm for $\ta_{i}$ starts at some $s_{i} \in S_{i}$ and, while the current strategy has at least one local improvement, moves to one of them, chosen arbitrarily. A run is the sequence of strategies visited, $s_{i}^{0} \to s_{i}^{1} \to \dots \to s_{i}^{k}$. A strategy with no local improvement is locally unimprovable.
\end{definition}

\begin{lemma}
\label{lem:terminate}
Under any neighbor relation, every run of the local improvement algorithm is finite and ends at a locally unimprovable strategy.
\end{lemma}

\begin{proof}
Each strategy in a run improves the one before it. The improvement relation is transitive, so if a run returned to a strategy it had visited, that strategy would improve itself, which is impossible. Hence no run visits a strategy twice, every run has at most $|S_{i}| - 1$ steps, and the last strategy of a run has no local improvement.
\end{proof}

\begin{definition}[Local improvement solvability]
\label{def:lis}
The mechanism $\Ga$ is local improvement solvable for $\ta_{i}$ under $N_{i}$ if all locally unimprovable strategies of $\ta_{i}$ are payoff-equivalent. It is local improvement solvable under $N = (N_{1}, \dots, N_{n})$ if it is local improvement solvable for every type $\ta_{i}$ of every player $i$ under $N_{i}$.
\end{definition}

Under the definition, every run of the algorithm stops at the same strategy up to payoff equivalence, from any starting strategy and with any choice of improvements along the way.\footnote{The definition does not require the locally unimprovable strategy to be unique. A strategy that is payoff-equivalent to a dominant strategy is itself dominant and therefore locally unimprovable, so uniqueness would rule out every mechanism in which some type has two dominant strategies.}

\begin{theorem}
\label{thm:lis}
(a) If $\Ga$ is local improvement solvable for $\ta_{i}$ under a neighbor relation $N_{i}$, its locally unimprovable strategies under $N_{i}$ are exactly the dominant strategies of $\ta_{i}$. Hence local improvement solvability under any neighbor relation implies strategy-proofness.

(b) Strategy-proofness does not imply local improvement solvability under one-step neighbors.
\end{theorem}

\begin{proof}
Part (a) follows from transitivity. Fix any strategy $s_{i}$. By Lemma~\ref{lem:terminate} a run from $s_{i}$ ends at a locally unimprovable strategy $p$, and each step of the run is a pointwise inequality, so $u_{i}(p, s_{-i}; \ta_{i}) \ge u_{i}(s_{i}, s_{-i}; \ta_{i})$ for every $s_{-i}$. Any locally unimprovable strategy $p'$ is payoff-equivalent to $p$, so $u_{i}(p', s_{-i}; \ta_{i}) \ge u_{i}(s_{i}, s_{-i}; \ta_{i})$ for every $s_{-i}$. Since $s_{i}$ was arbitrary, every locally unimprovable strategy is dominant. Conversely, no strategy improves a dominant strategy, so a dominant strategy has no local improvement and is locally unimprovable.

Part (b) uses the following counterexample. A single agent with one type chooses a strategy from $\{0, 1, 2\}$. There are no other players and no moves of nature, so there is a single contingency. Strategy $0$ pays her $1$, strategy $1$ pays $0$, and strategy $2$ pays $2$. One strategy thus improves another exactly when it pays more. Under one-step neighbors, as for bids in Section~\ref{sec:setup}, the only neighbor of $0$ is $1$, the neighbors of $1$ are $0$ and $2$, and the only neighbor of $2$ is $1$.

Strategy $2$ pays the most, so it is dominant and the mechanism is strategy-proof. No strategy improves a dominant strategy, so strategy $2$ is locally unimprovable. Strategy $0$ is locally unimprovable as well, since its only neighbor, strategy $1$, pays less. The locally unimprovable strategies $0$ and $2$ are not payoff-equivalent, since they pay $1$ and $2$. Hence the mechanism is not local improvement solvable under one-step neighbors.
\end{proof}

The example in part (b) shows the failure that local reasoning allows. A participant who compares strategy $0$ against its neighbor finds no improvement, although strategy $2$, two steps away, improves it.

The example has a single contingency, so the failure does not come from contingent reasoning. Strategy $2$ is even obviously dominant in the sense of Section~\ref{sec:extdef}, since with a single contingency the worst case and the best case of each strategy coincide. The failure comes instead from the neighbor relation, under which the best strategy is two steps from strategy $0$. With a single contingency, local improvement solvability is the condition that every local maximum of the payoff under the neighbor relation is a global maximum. The neighbor relation thus captures a source of simplicity beyond contingent reasoning, one that concerns the search through the strategy set. With three strategies the search is short, since the participant can compare every pair. Section~\ref{sec:sudoku} discusses Sudoku, in which the strategies are the many fillings of a grid and the same failure seems to make the puzzle difficult.

Local improvement solvability depends on the neighbor relation, and the dependence is monotone in the order of Definition~\ref{def:finer}.

\begin{proposition}
\label{prop:coarse}
(a) If $\Ga$ is local improvement solvable for $\ta_{i}$ under $N_{i}$ and $N_{i}$ is finer than $N_{i}'$, it is local improvement solvable for $\ta_{i}$ under $N_{i}'$.

(b) $\Ga$ is local improvement solvable for $\ta_{i}$ under the complete neighbor relation if and only if $\ta_{i}$ has a dominant strategy.
\end{proposition}

The proof is in Appendix~\ref{app:proofs}. Since every neighbor relation is finer than the complete one, the proposition orders the family of conditions by the neighbor relation, with strategy-proofness as its coarsest member. Strategy-proofness uses no structure on the strategy set, and under the complete neighbor relation the participant compares the candidate strategy against every other strategy. A finer neighbor relation uses more of the structure and gives a stronger condition. The strengthening can be strict, as the mechanism in Theorem~\ref{thm:lis}(b) shows.

Which neighbor relation should the analysis use? I treat the neighbor relation as part of the description of a mechanism, much as a type space is part of the description of an environment. In each application I take the relation that the presentation of the strategy set defines, which I call its natural neighbor relation. Under the natural neighbor relation, the neighbors of a strategy are the strategies one change away in its presentation, and for a bid they are the one-step neighbors.\footnote{The relations for the other mechanisms are defined where the mechanisms are introduced, in Sections~\ref{sec:classical} to~\ref{sec:dynamic}.} Choosing the natural neighbor relation makes local improvement solvability as demanding as the presentation allows and the comparison between two neighbors as small as it can be.

Appendix~\ref{app:remarks} explains the separate roles of the two requirements in Definitions~\ref{def:improves} and~\ref{def:lis}, that a local improvement be strictly better in some contingency and that all locally unimprovable strategies be payoff-equivalent.

\section{Classical mechanisms are local improvement solvable}
\label{sec:classical}

Are the classical strategy-proof mechanisms local improvement solvable? This section shows that five of them are, under their natural neighbor relations: the sealed-bid second-price auction, the multi-unit Vickrey auction, serial dictatorship, top trading cycles, and deferred acceptance. For the three matching mechanisms, the result holds in a rich environment, defined in Section~\ref{sec:matching}. In all five, the comparison between two neighboring strategies reduces to one case whose sign the participant determines from her own preference. The mechanisms are direct, apart from the grid of bids in the two auctions, and I call the recommended strategies truthful. All five are static mechanisms, in which each participant submits her strategy in one move and the outcome is computed from the strategies. Their dynamic counterparts are the subject of Section~\ref{sec:dynamic}.

\subsection{The sealed-bid second-price auction}
\label{sec:spa}

The sealed-bid second-price auction is defined as follows. There are $n \ge 2$ bidders and one object. Each bidder submits a bid on the grid $\{0, \dots, K\}$ with $K \ge 2$, simultaneously and without observing the others' bids. The object goes to the bidder with the highest bid, who pays the highest bid among the other bidders. If two or more bidders tie for the highest bid, the object is not sold and no one pays. A bidder who does not win pays nothing.

Fix one bidder, the agent. Her type is her value $\ta$ for the object, with $0 \le \ta < K$, which may or may not be a grid point. Her utility is her value minus her payment if she wins and zero otherwise. The contingency is the profile of the $n - 1$ rival bids. The agent's payoff depends on the rivals only through the highest rival bid $p$, and
\[
u(b, p; \ta) = (\ta - p)\,\one[b > p],
\]
where $\one[\cdot]$ is the indicator function. Let $T(\ta)$ be the set of grid points $b$ with $\ta \le b \le \ta + 1$, the truthful bids. If $\ta$ is not a grid point, $T(\ta)$ contains one bid, the least grid point above $\ta$. If $\ta$ is a grid point, $T(\ta) = \{\ta, \ta + 1\}$, and these two bids give the same utility in every contingency, since they give her different outcomes only when $p = \ta$. In that contingency the bid $\ta$ ties and loses, and the bid $\ta + 1$ wins at price $\ta$, both with payoff $0$.\footnote{This setup includes the design of \citet{li17} and \citet{bsk22}, in which values and bids are on the same grid and a tie for the highest bid means that no bidder wins.} The neighbor relation is one-step neighbors, so the neighbors of $b$ are $b - 1$ and $b + 1$ when these are on the grid.

\begin{proposition}
\label{prop:spa}
In the sealed-bid second-price auction, $u(b + 1, p; \ta) - u(b, p; \ta) = (\ta - b)\,\one[p = b]$ for adjacent bids. Under one-step neighbors, the locally unimprovable strategies are the bids in $T(\ta)$, and the auction is local improvement solvable.
\end{proposition}

The proposition follows directly from the payoff. Two neighboring bids give the agent different outcomes in exactly one case, the contingencies in which the highest rival bid equals the lower of the two bids. In that case the higher bid wins at that price and the lower bid loses. The difference is $\ta - b$, which the agent computes from her own value. It is positive for $b < \ta$, negative for $b > \ta$, and zero for $b = \ta$, the last case arising only when $\ta$ is a grid point.

Figure~\ref{fig:flow} shows the local improvements. An arrow from $b$ to $b'$ means that $b'$ locally improves $b$, $\De_{b} = u(b+1, p; \ta) - u(b, p; \ta)$ is the payoff difference between adjacent bids, and $a^{*}$ is the least grid point at or above $\ta$. The figure is drawn for a value that is not a grid point. Every bid below $a^{*}$ is improved by the bid one step above it, and every bid above $a^{*}$ by the bid one step below it. Thus $a^{*}$ is the only locally unimprovable strategy and, by Theorem~\ref{thm:lis}, the dominant bid. When $\ta$ is a grid point the arrow from $a^{*} + 1$ to $a^{*}$ is absent, and $a^{*}$ and $a^{*} + 1$ are two payoff-equivalent locally unimprovable strategies. In both cases the locally unimprovable strategies are the bids in $T(\ta)$, and every run of the algorithm ends at one of them.

The argument does not depend on $K$ or on the number $n$ of bidders, and the agent never compares a truthful bid against a bid other than its neighbors. By contrast, the sealed-bid second-price auction is not OSP, and a direct comparison of a truthful bid with a bid farther away requires her to check every highest rival bid between the two.

\begin{figure}[t]
\centering
\begin{tikzpicture}[x=1.42cm, y=1cm, >=Stealth, every node/.style={font=\small}]
\foreach \i/\l in {0/{$0$},1/{$1$},2/{$\cdots$},3/{$a^{*}\!-\!1$},4/{$a^{*}$},5/{$a^{*}\!+\!1$},6/{$\cdots$},7/{$K$}}
  \node (n\i) at (\i,0) {\l};
\foreach \a/\b in {0/1, 1/2, 2/3}
  \draw[->, thick] (n\a) to[bend left=42] (n\b);
\draw[->, very thick] (n3) to[bend left=42] (n4);
\foreach \a/\b in {7/6, 6/5}
  \draw[->, thick] (n\a) to[bend left=42] (n\b);
\draw[->, very thick] (n5) to[bend left=42] (n4);
\node at (1.5,0.98) {$\De_{b} = (\ta-b)\one[p=b]$, $\ta-b > 0$};
\node at (6.0,-0.98) {$\De_{b} = (\ta-b)\one[p=b]$, $\ta-b < 0$};
\end{tikzpicture}
\caption{Local improvements in the sealed-bid second-price auction, Proposition~\ref{prop:spa}.}
\label{fig:flow}
\end{figure}
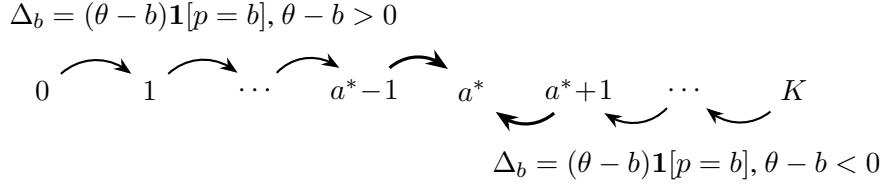

\subsection{Multi-unit demand and the Vickrey auction}
\label{sec:vickrey}

The multi-unit Vickrey auction is defined as follows. There are $Q$ identical units for sale and $n$ bidders. Each bidder submits a weakly decreasing vector of marginal bids on the grid $\{0, \dots, K\}$, one bid per unit she may receive, simultaneously and without observing the others' bids. I assume that any $n - 1$ of the bidders together submit at least $Q$ marginal bids. A bidder receives her $\ell$-th unit exactly when her $\ell$-th marginal bid exceeds the $(Q - \ell + 1)$-th highest marginal bid among the other bidders. Without ties the units go to the $Q$ highest marginal bids, and with ties some units can remain unsold, as the object does in the single-unit auction. A bidder who receives $x$ units pays, for her $\ell$-th unit, $\ell = 1, \dots, x$, the opportunity cost of that unit, the rival marginal bid that it displaces. The displaced bid is the $(Q - \ell + 1)$-th highest marginal bid among the other bidders. This payment rule is the Vickrey, or VCG, pricing rule for identical units, and it reduces to the second-price rule when $Q = 1$.

Fix the agent. She demands up to $L \le Q$ units. Her type is the vector $\ta = (\ta^{1}, \dots, \ta^{L})$ of marginal values, with $K > \ta^{1} > \dots > \ta^{L} \ge 0$, on or off the grid. Her utility is the sum of the marginal values of the units she receives minus her payments. Her strategy is a weakly decreasing vector $a = (a^{1}, \dots, a^{L})$ of marginal bids on the grid. For each unit $\ell$ let $T(\ta^{\ell})$ be the set of grid points $b$ with $\ta^{\ell} \le b \le \ta^{\ell} + 1$, as in the single-unit auction, and call a strategy truthful if $a^{\ell} \in T(\ta^{\ell})$ for every $\ell$. Since the marginal values are strictly decreasing, every truthful strategy is weakly decreasing. The contingency is the profile of the rivals' strategies. Under the pricing rule just stated, the agent's payoff depends on the rivals only through the vector of residual prices $p = (p_{1}, \dots, p_{L})$, where $p_{\ell}$ is the $(Q - \ell + 1)$-th highest rival bid, so that $p_{1} \le \dots \le p_{L}$. The residual prices are functions of the rivals' bids alone, and by the assumption above each can take every grid value. The agent wins unit $\ell$ if and only if $a^{\ell} > p_{\ell}$, and pays $p_{\ell}$ for it. Since $a$ is weakly decreasing and $p$ is weakly increasing, the units she wins are units $1$ through $x$ for some $x$. Her payoff is
\begin{equation}
u(a, p; \ta) \;=\; \sum_{\ell = 1}^{L} \big(\ta^{\ell} - p_{\ell}\big)\, \one[a^{\ell} > p_{\ell}],
\label{eq:unitsum}
\end{equation}
a sum of $L$ single-unit second-price payoffs. The neighbors of a strategy are its one-step neighbors on vectors, the weakly decreasing vectors that differ from it in one unit's bid by one grid step.

\begin{proposition}
\label{prop:vickrey}
In the multi-unit Vickrey auction, raising the bid for unit $\ell$ from $a^{\ell}$ to $a^{\ell} + 1$ changes the payoff by $(\ta^{\ell} - a^{\ell})\one[p_{\ell} = a^{\ell}]$. Under one-step neighbors on vectors, the locally unimprovable strategies are the truthful strategies, and the auction is local improvement solvable.
\end{proposition}

The proof is in Appendix~\ref{app:proofs}. The multi-unit problem consists of $L$ copies of the single-unit problem, connected only by the constraint that the strategy is weakly decreasing. The proof shows that from every strategy that is not truthful, the constraint leaves at least one local improvement toward a truthful strategy. As in the single-unit auction, the comparison between two neighboring strategies involves one case, the contingencies in which unit $\ell$'s residual price equals the lower bid. The sign of the payoff difference in that case is determined by unit $\ell$'s marginal value.\footnote{Identity (\ref{eq:unitsum}) holds under Vickrey pricing. It does not hold under uniform pricing, where the price of every unit depends on the agent's own strategy, so the argument does not apply to uniform-price auctions. The Ausubel auction of \citet{ausubel04} is a dynamic implementation of the Vickrey rule, and Section~\ref{sec:dynamic} discusses it.}

\subsection{Matching: local comparisons on rankings}
\label{sec:matching}

A matching mechanism assigns each of $n$ agents one of $m$ objects. Object $o$ has a capacity $q_{o} \ge 1$, the number of agents it can be assigned to, and $\sum_{o} q_{o} \ge n$, so that every agent is assigned an object. Each agent submits a strict ranking of all $m$ objects. Each object $o$ has a strict priority ranking $\succ_{o}$ over the agents. In the school choice interpretation the objects are schools with seats, and in the housing interpretation each object has capacity one and its owner has top priority at it. The three mechanisms studied are the following.

\emph{Serial dictatorship.} The agents are ordered by a single priority order, the same for every object, which may be drawn at random. The first agent in the order receives the object she ranks highest, the second receives the object she ranks highest among those with remaining capacity, and so on until every agent is assigned.

\emph{Top trading cycles.} In each round, every object with remaining capacity points to the remaining agent with the highest priority at it. Every remaining agent points to the object she ranks highest among those with remaining capacity. Since every remaining agent and every remaining object points at something, there is at least one cycle, and no two cycles intersect. Every agent in a cycle is assigned the object she points to and is removed, the capacity of that object falls by one, and the round ends. The mechanism stops when every agent is assigned. With capacities of one and owners at the top of the priorities this is the top trading cycles mechanism for housing markets of \citet{shapleyscarf74}.

\emph{Deferred acceptance.} In the first round every agent applies to the object she ranks highest. Each object tentatively holds the applicants with the highest priority at it, up to its capacity, and rejects the rest. In each later round every agent rejected in the previous round applies to the object she ranks highest among those that have not rejected her. Each object considers its new applicants together with the applicants it holds, again holds those with the highest priority up to its capacity, and rejects the rest. The mechanism stops in the first round with no rejection, and every agent is assigned the object that holds her. This is the agent-proposing, or student-proposing, deferred acceptance mechanism of \citet{galeshapley62}. Since every agent ranks every object and capacities suffice, every agent is assigned.

Fix an agent. Her type $\ta$ is her true strict ranking of the objects, and her strategy $s_{i}$ is a strict ranking. The contingency $s_{-i}$ is the profile of the others' strategies together with the priority data of the mechanism. The priority data are the priority order in serial dictatorship and the rankings $\succ_{o}$ in top trading cycles and deferred acceptance. Write $f(s_{i}, s_{-i})$ for the object she receives. The full strict domain is the type space in which every strict ranking is a possible type. The neighbors of a list are its one-exchange neighbors, the $m - 1$ lists obtained by exchanging two entries in adjacent positions.

\begin{definition}[Adjacent-swap property]
\label{def:swap}
The mechanism has the adjacent-swap property if, whenever $s_{i}$ and $s_{i}'$ differ by an adjacent exchange of $o$ and $o'$ with $s_{i}$ listing $o$ first, for every $s_{-i}$ either $f(s_{i}, s_{-i}) = f(s_{i}', s_{-i})$ or $f(s_{i}, s_{-i}) = o$ and $f(s_{i}', s_{-i}) = o'$. In the second case the pair $\{o, o'\}$ is decisive at $s_{-i}$ in both $s_{i}$ and $s_{i}'$.
\end{definition}

\begin{definition}[Rich environment]
\label{def:rich}
The environment is rich if every adjacent pair of every list is decisive at some contingency.
\end{definition}

Under the adjacent-swap property, two strategies that differ by an adjacent exchange of $o$ and $o'$ lead to different outcomes only in contingencies where the pair is decisive. There the strategy listing $o$ first receives $o$ and the other receives $o'$. The comparison between the two strategies therefore involves one case. The sign of the payoff difference in that case is determined by whether $\ta$ ranks $o$ above $o'$. The exchange toward $\ta$ is a local improvement if the pair is decisive in at least one contingency, which is what richness provides.

\begin{proposition}
\label{prop:matching}
Let strategies be strict rankings with one-exchange neighbors.

(a) If $f$ has the adjacent-swap property and the environment is rich, then $\ta$ is the only locally unimprovable strategy, and $f$ is local improvement solvable under one-exchange neighbors.

(b) Every deterministic mechanism on the full strict domain that is strategy-proof under truthful reporting has the adjacent-swap property.
\end{proposition}

The proof is in Appendix~\ref{app:proofs}. For part (a), a list other than $\ta$ has an adjacent pair ordered differently from $\ta$. Exchanging that pair is a local improvement by the adjacent-swap property and richness, and each exchange reduces the number of pairs ordered differently from $\ta$ by one, so every run ends at $\ta$. For part (b), strategy-proofness for the two types whose true rankings are the two lists forces any change in the outcome to be the exchange of the two objects.

\begin{corollary}
\label{cor:matching}
Serial dictatorship, top trading cycles, and student-proposing deferred acceptance have the adjacent-swap property, and in a rich environment they are local improvement solvable under one-exchange neighbors.
\end{corollary}

The corollary follows from Proposition~\ref{prop:matching}. The three mechanisms are deterministic, and on the full strict domain they are strategy-proof under truthful reporting for the agents. For deferred acceptance, strategy-proofness is shown by \citet{roth82matching} and \citet{dubinsfreedman81}, and for top trading cycles by \citet{roth82ttc} in housing markets and by \citet{abdulkadiroglusonmez03} with capacities and priorities. Part (b) therefore gives the adjacent-swap property. For serial dictatorship the property also follows from the rules directly. The earlier dictators do not read the agent's list, so the exchange of $o$ and $o'$ changes her outcome only when both are still available and she would have taken $o$. In a rich environment, part (a) then applies.

For example, the environment is rich when the agent's priority is part of the contingency, every profile of priorities and others' strategies is possible, and there are enough other agents. Appendix~\ref{app:rich} verifies richness under these conditions for each of the three mechanisms.

Richness depends on what the agent knows, and it fails when she knows that some adjacent pair on her list cannot affect her outcome. A list that is not payoff-equivalent to $\ta$ can then have no local improvement, as an example in Appendix~\ref{app:rich} shows.

In the random serial dictatorship treatments of \citet{li17}, every position remains possible whatever a subject's priority score, so the environment is rich and Corollary~\ref{cor:matching} applies. Appendix~\ref{app:lidesign} shows that the exchange of her two highest entries can be decisive whenever she does not move last, and this exchange is the most common error that \citet{li17} reports under list submission.

By Proposition~\ref{prop:matching}(b), the one-case comparison is available in every deterministic ranking mechanism on the full strict domain that is strategy-proof under truthful reporting. What differs across mechanisms is how the participant can learn that an adjacent exchange changes nothing but the two exchanged objects. In serial dictatorship this follows from the rules. In top trading cycles and deferred acceptance it is a consequence of strategy-proofness, and Section~\ref{sec:dynamicmatching} and Appendix~\ref{app:ida} argue that dynamic mechanisms display it.

\section{Extensive-form mechanisms and games: locality along situations}
\label{sec:extensive}

The neighbor relations of Section~\ref{sec:classical} are defined on the strategies of static mechanisms. This section treats extensive-form mechanisms formally, and with them extensive-form games. What is the neighbor relation there, and what is local improvement solvability under it? Section~\ref{sec:extdef} defines extensive-form mechanisms, one-situation neighbors, and four conditions, three existing and one introduced in this paper. Section~\ref{sec:hamming} relates local improvement solvability under one-situation neighbors to strict action-proofness, and Section~\ref{sec:taxonomy} relates the four conditions to one another.

\subsection{Extensive-form mechanisms and simplicity concepts}
\label{sec:extdef}

Consider a finite extensive-form mechanism of perfect recall, a game tree with an outcome function, with nodes defined as action sequences as in \citet{osborne94}. Let $I$ be the set of players, including nature as the nonstrategic player $0$, and let $I_{+} = I \setminus \{0\}$ be the set of strategic players. I follow the terminology of \citet{watsonzhang}. A situation is a place in the game where a player must choose an action. I reserve the term information set for the information about the strategy profile that the player has at that place, expressed as a set of strategy profiles. Let $H$ be the set of situations, let $H_{i}$ be the set of situations of player $i \in I$, and for $h \in H$ let $A(h)$ be the set of actions available at $h$, with $|A(h)| \ge 2$. A strategy profile is $s = (a_{h})_{h \in H}$ with $a_{h} \in A(h)$, and the set of strategy profiles is $S = \times_{h \in H} A(h)$. For $L \subset H$, the profile of actions at the situations in $L$ is $s_{L} = (a_{h})_{h \in L}$, with $S_{L}$ the corresponding projection of $S$ and $-L = H \setminus L$. Player $i$'s strategy is $s_{i} = s_{H_{i}}$, and $S_{i} = \times_{h \in H_{i}} A(h)$ is her strategy set. A strategy specifies an action at every situation in $H_{i}$, including situations that the strategy itself makes unreachable. For $h \in H_{i}$, $s_{h}$ is the action that $s_{i}$ specifies at $h$. Nature's strategy $s_{0}$ is a component of the profile, and no probability distribution over nature's strategy is specified. A contingency for player $i \in I_{+}$ is a profile $s_{-i} \in S_{-i}$ of the other players' strategies, nature's included. Player $i$ has a type $\ta_{i} \in \Ta_{i}$ and a payoff function $u_{i} : S \times \Ta_{i} \to \R$, written $u_{i}(s_{i}, s_{-i}; \ta_{i})$, which is her utility from the outcome at the terminal node that the profile reaches.

The extensive-form mechanism is a mechanism in the sense of Section~\ref{sec:setup}, with the strategy sets $S_{i}$ and the outcome function given by the terminal nodes. Its reduced normal-form mechanism is the mechanism obtained by identifying two strategies of a player whenever they reach the same terminal node in every contingency. A static mechanism is an extensive-form mechanism in which every strategic player has exactly one situation, and a dynamic mechanism is one in which some strategic player has more than one. When the payoffs are fixed rather than indexed by types, the same object is called a game. An extensive-form game is therefore the special case of an extensive-form mechanism in which every player has a single type.

Conversely, fixing a type $\ta_{i}$ of player $i$ turns an extensive-form mechanism into an extensive-form game for her, with payoffs $u_{i}(\cdot\,; \ta_{i})$ and with the other players and nature as the source of her contingencies. Each condition in this section, and local improvement solvability, is stated for a single type and quantifies only over that type's strategies and the contingencies she faces. A mechanism therefore has each property exactly when every game that a type faces has it. The definitions and results below apply without change to extensive-form games, and the examples that separate the conditions are games of this kind.

For $h \in H$, let $S(h)$ be the set of strategy profiles that reach $h$, that is, the profiles whose induced action sequence has a node in $h$ as a truncation. $S(h)$ is the information set at $h$. For $i \in I_{+}$ and $h \in H_{i}$, perfect recall implies that $S(h)$ is a restriction for player $i$, $S(h) = S(h)_{i} \times S(h)_{-i}$. A strategy $s_{i}$ reaches $h$ if $s_{i} \in S(h)_{i}$, and a contingency $s_{-i}$ reaches $h$ if $s_{-i} \in S(h)_{-i}$. A situation $h \in H_{i}$ precedes $h' \in H_{i}$, and $h'$ follows $h$, if every node in $h'$ has a node in $h$ as a strict prefix. Perfect recall implies that for every $h' \in H_{i}$ that precedes $h$ the projection $S(h)_{h'}$ is a single action, the action at $h'$ that leads toward $h$. Perfect recall also implies that for every other $h' \in H_{i}$ the projection $S(h)_{h'}$ is $A(h')$. Hence $S(h)_{i}$ is the set of strategies that specify, at every situation of player $i$ preceding $h$, the action that leads toward $h$. Two strategies that both reach $h$ therefore agree at every situation preceding $h$. For a strategy $r \in S_{i}$, a situation $h \in H_{i}$, and an action $a \in A(h)$, write $r[h \leftarrow a]$ for the strategy $(a, r_{H_{i} \setminus \{h\}})$ that agrees with $r$ at every situation other than $h$ and specifies $a$ at $h$. A continuation of an action $a$ at $h$ is a strategy $r \in S(h)_{i}$ with $r_{h} = a$. A situation $h'$ follows $h$ through $a$ if $h'$ follows $h$ and $a$ is the action at $h$ that leads toward $h'$, so that $S(h')_{h} = \{a\}$. The actions that $r$ specifies at the situations that follow $h$ through $a$ are its continuation after $a$.

The information in the extensive-form structure already restricts the comparisons that a player must make. At a situation $h$ she knows that play has reached $h$, so the strategies she compares there are those in $S(h)_{i}$. Every strategy that reaches $h$ is then in effect a neighbor of every other strategy that reaches $h$. The contingencies are reduced in the same way. Suppose two strategies differ only at $h$, or only at $h$ and at situations following $h$. In every contingency that does not reach $h$ they lead to the same terminal node, since play never passes through a situation at which they differ. A comparison between the two therefore involves only the contingencies in $S(h)_{-i}$. The extensive-form mechanism shows the player both restrictions. Its reduced normal-form mechanism presents the same strategies as wholes and shows neither. When $S(h)_{-i}$ is smaller than $S_{-i}$, the dominant strategy can therefore be easier to find in an extensive-form mechanism than in its reduced normal-form mechanism. The conditions defined below all use these restrictions.

A finer relation can also be used. The one-situation neighbors of $s_{i}$ are the strategies that differ from $s_{i}$ at exactly one situation,
\[
N_{i}(s_{i}) \;=\; \big\{ s_{i}' \in S_{i} : \big|\{h \in H_{i} : s_{h} \neq s_{h}'\}\big| = 1 \big\}.
\]
Under this relation, a step of the local improvement algorithm changes the action at one situation and holds the rest of the strategy fixed. The one-situation neighbors of $s_{i}$ are thus its one-shot deviations $s_{i}[h \leftarrow a]$ with $h \in H_{i}$ and $a \in A(h) \setminus \{s_{h}\}$, the strategies against which the one-shot deviation principle checks $s_{i}$.\footnote{The principle is presented, for instance, in \citet{osborne94}, and \citet{hendonetal96} extend it from subgame perfection to sequential rationality.} Section~\ref{sec:hamming} shows what local improvement solvability then requires, and Section~\ref{sec:taxonomy} ends by comparing it with the principle. One-situation neighbors differ in a single action, the smallest change of a strategy that the extensive-form structure defines.\footnote{A mechanism whose strategy set has more structure can support a relation with an even simpler path to the dominant strategy, as the ascending clock auction of Section~\ref{sec:clock} shows.}

The following definitions state three existing concepts, namely dominance, the obvious dominance of \citet{li17}, and the strong obvious dominance of \citet{pyciatroyan23}. The definitions also state a fourth condition, action dominance, which I introduce in this paper. Fix a player $i \in I_{+}$, a type $\ta_{i}$, and a strategy $s_{i}^{*}$, called the recommended strategy. Under a recommendation $\si$ in the sense of Section~\ref{sec:setup}, $s_{i}^{*} = \si_{i}(\ta_{i})$. A strategy $s_{i}'$ deviates from $s_{i}^{*}$ at a situation $h \in H_{i}$ if both strategies reach $h$, that is, $s_{i}^{*}, s_{i}' \in S(h)_{i}$, and $s_{h}' \neq s^{*}_{h}$. Recall that two strategies that reach $h$ agree at every situation preceding $h$. Hence the situations at which $s_{i}'$ deviates from $s_{i}^{*}$ are those at which the two strategies first differ along some path of play, which \citet{li17} calls their earliest points of departure. Each condition below is stated at a situation $h$ that $s_{i}^{*}$ reaches, and the contingencies in each comparison are those that reach $h$, the elements of $S(h)_{-i}$. Each condition depends on $h$ through $S(h)_{-i}$ and through the set $S(h)_{i}$ of the player's own strategies that reach $h$, which in action dominance and strong obvious dominance supplies both sides of the comparison.

\begin{definition}[Dominance at a situation]
\label{def:dominant}
The strategy $s_{i}^{*}$ is dominant for $\ta_{i}$ if (\ref{eq:SPdef}) holds. Equivalently, at every situation $h$ that $s_{i}^{*}$ reaches and for every $s_{i}'$ that deviates from $s_{i}^{*}$ at $h$,
\begin{equation}
u_{i}(s_{i}^{*}, s_{-i}; \ta_{i}) \;\ge\; u_{i}(s_{i}', s_{-i}; \ta_{i}) \qquad \text{for every } s_{-i} \in S(h)_{-i}.
\tag{D}
\label{eq:D}
\end{equation}
\end{definition}

\begin{definition}[Obvious dominance]
\label{def:od}
The strategy $s_{i}^{*}$ is obviously dominant for $\ta_{i}$ if, at every situation $h$ that $s_{i}^{*}$ reaches and for every $s_{i}'$ that deviates from $s_{i}^{*}$ at $h$,
\begin{equation}
\min_{s_{-i} \in S(h)_{-i}} u_{i}(s_{i}^{*}, s_{-i}; \ta_{i}) \;\ge\; \max_{s_{-i}' \in S(h)_{-i}} u_{i}(s_{i}', s_{-i}'; \ta_{i}).
\tag{OD}
\label{eq:OD}
\end{equation}
\end{definition}

\begin{definition}[Action dominance]
\label{def:ad}
The strategy $s_{i}^{*}$ is action dominant for $\ta_{i}$ if, at every situation $h$ that $s_{i}^{*}$ reaches,
\begin{equation}
\begin{split}
& u_{i}(r[h \leftarrow s^{*}_{h}], s_{-i}; \ta_{i}) \;\ge\; u_{i}(r[h \leftarrow a'], s_{-i}; \ta_{i}) \\
& \qquad \text{for every } r \in S(h)_{i}, \; a' \in A(h), \text{ and } s_{-i} \in S(h)_{-i}.
\end{split}
\tag{AD}
\label{eq:AD}
\end{equation}
It is strictly action dominant if, in addition, for every such $h$, $r$, and $a' \neq s^{*}_{h}$ the inequality is strict for some $s_{-i} \in S(h)_{-i}$.
\end{definition}

\begin{definition}[Strong obvious dominance]
\label{def:sod}
The strategy $s_{i}^{*}$ is strongly obviously dominant for $\ta_{i}$ if, at every situation $h$ that $s_{i}^{*}$ reaches and for every $a' \in A(h)$ with $a' \neq s^{*}_{h}$,
\begin{equation}
\min_{\substack{r \in S(h)_{i} : \, r_{h} = s^{*}_{h} \\ s_{-i} \in S(h)_{-i}}} u_{i}(r, s_{-i}; \ta_{i}) \;\ge\; \max_{\substack{r' \in S(h)_{i} : \, r'_{h} = a' \\ s_{-i}' \in S(h)_{-i}}} u_{i}(r', s_{-i}'; \ta_{i}).
\tag{SOD}
\label{eq:SOD}
\end{equation}
\end{definition}

The mechanism is strategy-proof (SP), obviously strategy-proof (OSP), or action-proof (AP) if every type of every strategic player has a dominant, an obviously dominant, or an action dominant strategy, respectively. It is strictly action-proof or strongly obviously strategy-proof (SOSP) if every type of every strategic player has a strictly action dominant or a strongly obviously dominant strategy, respectively.\footnote{For two such properties, $X \Rightarrow Y$ means that every mechanism with $X$ has $Y$, $X \not\Rightarrow Y$ that some mechanism with $X$ lacks $Y$, and $X + Y$ is the conjunction. The same notation is used for the conditions (\ref{eq:D}), (\ref{eq:OD}), (\ref{eq:AD}), and (\ref{eq:SOD}) on a strategy.}

Condition (\ref{eq:D}) is dominance stated situation by situation. Condition (\ref{eq:OD}) is the definition of \citet{li17}, and condition (\ref{eq:SOD}) is the definition of \citet{pyciatroyan23}. Compare the two definitions. In (\ref{eq:OD}) the deviation is a strategy, and the comparison holds the recommended strategy fixed after $h$. In (\ref{eq:SOD}) the deviation is an action, and the comparison ranges over everything the player might do after $h$ on either side. These three conditions compare payoffs across the contingencies in $S(h)_{-i}$, either one contingency at a time or by the worst case against the best case, and none uses a neighbor relation on the strategy set. Conditions (\ref{eq:OD}) and (\ref{eq:SOD}) make the second comparison and are worst-best comparison concepts in the sense of the introduction.

Action dominance takes the deviation of (\ref{eq:SOD}), an action, and the comparison of (\ref{eq:D}), one contingency at a time.\footnote{One precedent is the weak dominance between actions of \citet{catoninixue}, which Section~\ref{sec:taxonomy} compares with action dominance.} The player's strategy at her other situations is arbitrary but the same on both sides of the comparison. Action dominance also has a precedent in the agent normal form of \citet{selten75}, in which each situation of a player is treated as a separate agent who chooses the action there. Action dominance at $h$ is weak dominance of the recommended action for the agent at $h$, with the player's own actions at her other situations among that agent's contingencies. This weak dominance is required at every situation that $s_{i}^{*}$ reaches. Unlike conditions (\ref{eq:D}), (\ref{eq:OD}), and (\ref{eq:SOD}), action dominance compares two strategies that differ at one situation only. A change of the action at one situation is a move to a one-situation neighbor, a one-shot deviation. Theorem~\ref{prop:hamming} shows that, in extensive-form mechanisms without duplicate dominant strategies, local improvement solvability under one-situation neighbors is equivalent to strict action-proofness.

Each of the four conditions is stated situation by situation, and Section~\ref{sec:taxonomy} and Appendix~\ref{app:other} use the version at a single situation. Fix a situation $h$ that $s_{i}^{*}$ reaches and an action $a' \in A(h)$ with $a' \neq s^{*}_{h}$. The strategy $s_{i}^{*}$ satisfies dominance at $h$ against $a'$ if (\ref{eq:D}) holds for every $s_{i}' \in S(h)_{i}$ with $s_{h}' = a'$. It satisfies obvious dominance at $h$ against $a'$ if (\ref{eq:OD}) holds for every such $s_{i}'$. It satisfies action dominance at $h$ against $a'$ if (\ref{eq:AD}) holds for every $r \in S(h)_{i}$ and every $s_{-i} \in S(h)_{-i}$. It satisfies strong obvious dominance at $h$ against $a'$ if (\ref{eq:SOD}) holds. Appendix~\ref{app:dominance} proves the equivalence of (\ref{eq:SPdef}) and (\ref{eq:D}) in Definition~\ref{def:dominant}.

The last concept is the complexity measure of \citet{nagelsaitto}. It first merges consecutive situations of the player between which no payoff-relevant information arrives. It then removes, one at a time, actions whose best payoff at a situation, over everything that can follow, is no larger than the worst payoff of another surviving action there. Finally, it counts the actions that survive, and the complexity is the largest such count over the players, their types, and their situations.\footnote{Appendix~\ref{app:ns} gives a simplified version of the construction in the notation of this paper and defines the terms used in Section~\ref{sec:games}, such as the evident reduction and the essential frame.} The results in Section~\ref{sec:games} use only two properties of the measure. The removal step compares the maximum of one set of payoffs against the minimum of another. Two consecutive situations of the player are merged unless the information between them excludes some payoff possible under her dominant continuation.

\subsection{One-situation neighbors}
\label{sec:hamming}

A change of action at a situation that the current strategy does not reach changes no outcome, so it is not a local improvement. A change of action at a situation $h$ that the current strategy reaches replaces the strategy $r$ by $r[h \leftarrow a]$. The comparison between the two strategies is the comparison in (\ref{eq:AD}) with $r$ fixed, together with the requirement that the change be strictly better in some contingency. Strict action dominance, Definition~\ref{def:ad}, is exactly (\ref{eq:AD}) together with that requirement.

Two dominant strategies of a type are duplicates if they specify different actions at some situation that both of them reach. Two dominant strategies are payoff-equivalent, so duplicates exist only when two different actions at some situation lead to continuations that are equally good in every contingency, as in the game of Figure~\ref{fig:duplicate}. In a static mechanism every strategy reaches the player's single situation, so any two different dominant strategies are duplicates. The truthful bids $\ta$ and $\ta + 1$ of the sealed-bid second-price auction with a value on the bid grid are an example.

\begin{figure}[t]
\centering
\begin{tikzpicture}[x=1cm, y=1cm, every node/.style={font=\small}, dn/.style={circle, fill, inner sep=1.4pt}]
\node[dn, label=above:{$i$}] (r) at (0,0) {};
\node[dn, label=left:{$i$}] (l) at (-1.6,-1.2) {};
\node[dn, label=right:{$i$}] (m) at (1.6,-1.2) {};
\node (lx) at (-2.3,-2.4) {$10$};
\node (ly) at (-0.9,-2.4) {$0$};
\node (mx) at (0.9,-2.4) {$10$};
\node (my) at (2.3,-2.4) {$0$};
\draw (r) -- node[above left] {$a^{*}$} (l);
\draw (r) -- node[above right] {$a'$} (m);
\draw (l) -- node[left] {$x$} (lx);
\draw (l) -- node[right] {$y$} (ly);
\draw (m) -- node[left] {$x'$} (mx);
\draw (m) -- node[right] {$y'$} (my);
\end{tikzpicture}
\caption{A game with duplicate dominant strategies, Theorem~\ref{prop:hamming}.}
\label{fig:duplicate}
\end{figure}
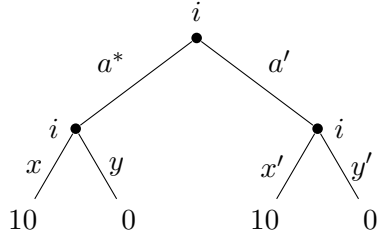

\begin{theorem}
\label{prop:hamming}
The type $\ta_{i}$ has a strictly action dominant strategy if and only if $\Ga$ is local improvement solvable for $\ta_{i}$ under one-situation neighbors and $\ta_{i}$ has no duplicate dominant strategies.
\end{theorem}

The proof is in Appendix~\ref{app:proofs}. If $s_{i}^{*}$ is strictly action dominant, any strategy that differs from it at a situation both reach can be improved there. Every locally unimprovable strategy is therefore payoff-equivalent to $s_{i}^{*}$, and a similar argument rules out duplicates. Conversely, let $s_{i}^{*}$ be dominant and suppose strict action dominance fails for it at a situation $h$. A run of the algorithm from the failing deviation at $h$, with the actions after $s^{*}_{h}$ held fixed, ends at a locally unimprovable strategy that does not play $s^{*}_{h}$ at $h$. Under local improvement solvability this strategy is dominant and differs from $s_{i}^{*}$ at a situation both reach, so the two are duplicates.

Consider the dynamic serial dictatorship of Section~\ref{sec:dynamicmatching}, in which the agents are called one at a time in the priority order and each takes one of the objects that remain. Each agent's choice from the offered set determines her outcome, so choosing her most preferred object in that set is strictly action dominant. By the theorem, the dynamic serial dictatorship is local improvement solvable under one-situation neighbors. By Corollary~\ref{cor:matching}, its direct mechanism, serial dictatorship with list submission, is local improvement solvable under one-exchange neighbors in a rich environment.

Local improvement solvability alone does not rule out duplicates. In the game of Figure~\ref{fig:duplicate}, the agent chooses $a^{*}$ or $a'$, and the two choices lead to identical subgames. After $a^{*}$ she chooses $x$, paying $10$, or $y$, paying $0$, and after $a'$ she chooses $x'$, paying $10$, or $y'$, paying $0$. Here and in Section~\ref{sec:taxonomy}, a strategy is written as the list of the agent's actions, from the top of the tree down and from left to right. A dot marks a situation whose action does not matter. Every strategy with payoff $0$ is improved by correcting its action at the second situation it reaches, so every locally unimprovable strategy pays $10$ and the game is local improvement solvable under one-situation neighbors. The strategies $(a^{*}, x, \cdot)$ and $(a', \cdot, x')$ are both dominant and differ at the first situation, which both reach, so they are duplicates. The game is not action-proof. If $a^{*}$ is recommended, the strategy $r = (\cdot, y, x')$ pays $0$ with $a^{*}$ and $10$ with $a'$, and the case with $a'$ recommended is symmetric.

\subsection{Four conditions on a step}
\label{sec:taxonomy}

The four conditions of Definitions~\ref{def:dominant} to~\ref{def:sod} differ in two respects. The first is whether the player's own actions after the recommended action are held equal to those of $s_{i}^{*}$ or allowed to vary. The second is whether the contingencies are compared one at a time or by comparing the worst case of one side against the best case of the other. Table~\ref{tab:four} arranges the four conditions. Its rows are the two ways of comparing contingencies and its columns are the two treatments of the player's own other actions.

Moving down a column replaces the comparison contingency by contingency with the comparison of the worst case against the best case. This replacement is the difference between strategy-proofness and OSP and between action-proofness and SOSP, and it removes the need to consider which contingency obtains. Moving across a row allows the player's own actions at other situations to vary. This variation removes the need to consider what the player will do at later situations. Action-proofness is the condition obtained by moving across a row alone.

Most other worst-best comparison concepts lie between these four conditions. The partition-based concepts of \citet{zhanglevin21} and \citet{arribillagamassoneme24} fill the left column of Table~\ref{tab:four}, between dominance and obvious dominance, and the concepts of \citet{pyciatroyan23} fill the bottom row. The conditions of \citet{catoninixue} compare two actions and treat the player's own continuation in a third way. Appendix~\ref{app:other} relates all of them to the table.

Comparing actions is not what places a condition in the right column, since the conditions of \citet{catoninixue} compare actions as well. What places it there is the requirement that the comparison hold for every continuation after the recommended action. Action dominance is the entry of that column in which the contingencies are compared one at a time. As a property of mechanisms it does not appear in the literature as far as I know, although each neighboring condition in the table does. The weak dominance between actions of \citet{catoninixue} is a closely related precedent. It asks only that some continuation after the recommended action match each continuation after $a'$. Weak dominance between actions, required at every situation that a strategy reaches, is equivalent to dominance. Hence the requirement for every continuation after the recommended action is exactly what separates action-proofness from strategy-proofness. I know of no concept between strategy-proofness and action-proofness, or between action-proofness and SOSP.

\begin{figure}[t]
\centering
\begin{tikzpicture}[x=1cm, y=1cm, every node/.style={font=\small}, dn/.style={circle, fill, inner sep=1.4pt}]
\node[dn, label=above:{$i$}] (r) at (0,0) {};
\node[dn, label=left:{$i$}] (l) at (-1.6,-1.2) {};
\node[dn, label=right:{$i$}] (m) at (1.6,-1.2) {};
\node (lx) at (-2.3,-2.4) {$10$};
\node (ly) at (-0.9,-2.4) {$0$};
\node (mx) at (0.9,-2.4) {$0$};
\node (my) at (2.3,-2.4) {$5$};
\draw (r) -- node[above left] {$a^{*}$} (l);
\draw (r) -- node[above right] {$a'$} (m);
\draw (l) -- node[left] {$x$} (lx);
\draw (l) -- node[right] {$y$} (ly);
\draw (m) -- node[left] {$x'$} (mx);
\draw (m) -- node[right] {$y'$} (my);
\end{tikzpicture}
\caption{The two-move game, Proposition~\ref{prop:taxonomy}(b).}
\label{fig:twomove}
\end{figure}
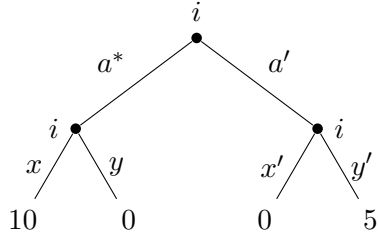

\begin{table}[t]
\centering
\small
\begin{tabular}{lcc}
\hline
 & own other actions equal to $s_{i}^{*}$ & own other actions arbitrary \\
\hline
contingency by contingency & strategy-proof & action-proof \\
worst case against best case & OSP & SOSP \\
\hline
\end{tabular}
\caption{Four conditions on a step of the local improvement algorithm at a situation.}
\label{tab:four}
\end{table}

\begin{proposition}
\label{prop:taxonomy}
(a) (\ref{eq:SOD}) $\Rightarrow$ (\ref{eq:OD}) $\Rightarrow$ (\ref{eq:D}) and (\ref{eq:SOD}) $\Rightarrow$ (\ref{eq:AD}) $\Rightarrow$ (\ref{eq:D}). Hence $\mathrm{SOSP} \Rightarrow \mathrm{OSP} \Rightarrow \mathrm{SP}$ and $\mathrm{SOSP} \Rightarrow \mathrm{AP} \Rightarrow \mathrm{SP}$.

(b) No implication in (a) reverses, $\mathrm{OSP} \not\Rightarrow \mathrm{AP}$, and $\mathrm{AP} \not\Rightarrow \mathrm{OSP}$.

(c) $\mathrm{OSP} + \mathrm{AP} \not\Rightarrow \mathrm{SOSP}$.
\end{proposition}

Parts (a) and (c) together say that obvious strategy-proofness and action-proofness are both necessary for strong obvious strategy-proofness, and that the two together are not sufficient for it. The proof is in Appendix~\ref{app:proofs}. All four implications in part (a) are immediate from the definitions, situation by situation.

Parts (b) and (c) are proved by four examples. The two that show that obvious strategy-proofness and action-proofness are independent are given here, and the other two are in Appendix~\ref{app:proofs}. In the game of Figure~\ref{fig:twomove} no other player moves, so there is a single contingency, and strategies are written as in Section~\ref{sec:hamming}.

The game of Figure~\ref{fig:twomove} is OSP and not AP. The recommended strategy $(a^{*}, x, \cdot)$ pays $10$, the best payoff in the game, so it is obviously dominant. It is not action dominant, since for the strategy $r = (\cdot, y, y')$ the first action $a^{*}$ pays $0$ and $a'$ pays $5$. The recommended action at the first situation is best only if the agent goes on to play $x$. Obvious dominance holds this later action fixed, and action dominance does not.

The sealed-bid second-price auction is AP and not OSP. In a static mechanism each player has one situation, at which action dominance and dominance are the same condition, so the auction is action-proof, and \citet{li17} shows that it is not OSP.

Theorem~\ref{prop:hamming} places local improvement solvability under one-situation neighbors next to action-proofness in Table~\ref{tab:four}. Absent duplicate dominant strategies, the two differ only by strictness, so local improvement solvability under one-situation neighbors relates to obvious strategy-proofness as action-proofness does. The sealed-bid second-price auction with a value off the grid is strictly action-proof and not OSP, while the two-move game of Figure~\ref{fig:twomove} is OSP and not action-proof. Strong obvious strategy-proofness implies action-proofness but not local improvement solvability under one-situation neighbors. In dynamic Sudoku, Section~\ref{sec:sudoku}, the recommended digit at a cell is action dominant without being strictly so, since an error at a later cell can leave the payoff at $0$ whichever digit is entered. This lack of strictness is the only way a strongly obviously strategy-proof mechanism can fail to be local improvement solvable under one-situation neighbors. The two-move game of Figure~\ref{fig:twomove} shows the other way a second locally unimprovable strategy arises. The strategy $(a', y, y')$ pays $5$ and its three one-situation neighbors pay $0$, $5$, and $0$, so it can be improved only by changing $a'$ to $a^{*}$ and $y$ to $x$ at once.

Theorem~\ref{prop:hamming} can be read as a one-shot deviation principle for dominance. The classical principle fixes the strategies of the others and checks the one-shot deviations from the candidate strategy at every situation, as if play had reached it. In a finite game, a strategy that passes every such check is optimal. Action dominance requires that, at each situation that $s_{i}^{*}$ reaches, no one-shot deviation from the recommended action pay more, whatever the player does at her other situations and whatever the others do. When the type has a strictly action dominant strategy, every strategy with no improving one-shot deviation is dominant, by the theorem and Theorem~\ref{thm:lis}(a). Absent duplicate dominant strategies, the converse holds as well. The requirement on the player's other actions matters because the local improvement algorithm may start from a strategy that errs at several situations, such as $(a', y, y')$.

\section{Three games and the worst-best comparison concepts}
\label{sec:games}

This section examines the worst-best comparison concepts through three games. The pick-a-number game, which a short argument solves, is not OSP and is maximally complex in the measure of \citet{nagelsaitto}. Chess, when White has a winning strategy, is OSP for White, and Sudoku with a unique solution is SOSP, although both seem difficult for a participant. Local improvement solvability under the natural neighbor relation of each game accounts for part of these discrepancies, since it holds in the easy game and fails in the two difficult ones.

\subsection{The pick-a-number game}
\label{sec:additive}

Two players each pick a number from $A = \{0, \dots, K\}$, and player $i$ receives $u_{i}(a_{i}, a_{j}) = a_{i} + \lambda a_{j}$, where $a_{j}$ is the opponent's number and $\lambda > K$. There is no private information, so the type is degenerate and omitted. The one-step neighbors of a number are the numbers one step above and below it. The classifications below hold the payoffs and the set $A$ fixed.

\begin{proposition}
\label{prop:additive}
In the pick-a-number game with $K \ge 1$ and $\lambda > K$:

(a) $u_{i}(a + 1, a_{j}) - u_{i}(a, a_{j}) = 1$ for every $a < K$ and every $a_{j}$, and the game is local improvement solvable under one-step neighbors.

(b) The static game is not OSP, and its complexity in the measure of \citet{nagelsaitto} is $K + 1$.

(c) Let an SP extensive-form game implement the outcomes of the static game through maps $\phi$, in the sense of Definition~\ref{def:implements}. For numbers $a \neq a'$ and $b$, let $h$ be the first situation of player $i$ on the play of $(\phi_{i}(a), \phi_{j}(b))$ at which $\phi_{i}(a)$ and $\phi_{i}(a')$ differ. Then every $\phi_{j}(b')$ reaches $h$. No such game is OSP. If each $\phi_{i}(a)$ gives $i$ the number $a$ whatever $j$ does, the complexity is at least $K + 1$, and every strategy of $j$ reaches every situation of $i$ that two such strategies reach.
\end{proposition}

Parts (a) and (b) are immediate, and part (c) is proved in Appendix~\ref{app:proofs}. In part (a), the term $\lambda a_{j}$ cancels from the difference between two neighboring numbers, so the difference is $1$ in every contingency. Hence $a + 1$ locally improves $a$ for every $a < K$, and $K$ is the only locally unimprovable strategy. The participant does not need to know which contingency obtains. Part (a) is a statement about payoff differences. The worst-best comparison concepts compare payoff levels. In part (b) the minimum payoff from $K$ is $K$, at $a_{j} = 0$, and the maximum payoff from any other number $a$ is $a + \lambda K > K$, at $a_{j} = K$, so (\ref{eq:OD}) fails. Since $a + \lambda K$ exceeds every number in $A$, no number is evidently never better than another, and the evident reduction is $A$.

For part (c), an extensive-form game can reveal a player's earlier actions to the other player, and the other player can condition on them. Since $\lambda > K$, such conditioning lets a player punish the opponent's recommended strategy. So wherever two of a player's strategies that represent numbers first differ on a play against such a strategy of the opponent, strategy-proofness requires every such strategy of the opponent to reach that situation. Along these strategies every strategy-proof implementation thus works like the static game, and (\ref{eq:OD}) fails where those for $K$ and a smaller number first differ. The measure of \citet{nagelsaitto} also registers information about moves that represent no number, such as a veto. Suppose instead that each of a player's strategies that represents a number gives her that number whatever the opponent does. She then cannot tell two strategies of the opponent apart at a situation that two of these strategies reach, for the same reason. These situations therefore merge into one, so dynamic implementation cannot lower the complexity, although the assumption itself restricts no information.

The parameter $\lambda$ shows how the classification by the worst-best comparison concepts depends on the weight of the opponent's number. For $\lambda = 1$ there is a dynamic game that implements the outcomes of the static game. In it the players alternately raise their numbers by one step, observe each other, and cannot resume raising once they stop. That game is obviously strategy-proof, because a number already committed bounds the loss from any later action of the opponent. For $\lambda > K$ that bound never suffices, and no strategy-proof extensive-form game that implements these outcomes is OSP. A classification by obvious strategy-proofness that changes with the exchange rate $\lambda$, while the decision is the same, seems to measure something other than the difficulty of the decision.

\subsection{Chess}
\label{sec:chess}

Chess and Sudoku show the opposite error. Consider chess with payoff $1$ for a win, $1/2$ for a draw, and $0$ for a loss, and suppose, following \citet{pyciatroyan23}, that White has a winning strategy.\footnote{Black need not have a dominant strategy, so the classifications of chess in this paper concern White.} Any winning strategy is then obviously dominant, since its worst outcome at every position it reaches is a win and no deviation does better. It is not strongly obviously dominant, since after the winning move some continuation of White loses while some deviation wins against an error of Black.\footnote{\citet{pyciatroyan23} use this observation to motivate their stronger concepts.} Under one-situation neighbors, however, chess is not local improvement solvable for White. Call a move losing if Black can prevent a win after it. Fix a winning strategy, a position $p$ that it reaches, its move $w$ there, and a losing move $m$ at $p$ after which White can still win against some error of Black. Take a strategy that makes $m$ at $p$ and a losing move at a position $q$ that the winning strategy reaches after $w$. Appendix~\ref{app:chess} specifies the winning strategy and the other moves so that no other change at one position improves the strategy. Correcting $m$ alone is not a local improvement, because the corrected strategy reaches $q$. The corrected strategy fails to win against a strategy of Black that leads to $p$, errs after $m$ so that the original strategy wins, and after $w$ leads to $q$ and prevents a win. Correcting the move at $q$ alone changes nothing, since $q$ is not reached. The strategy with the two losing moves is therefore locally unimprovable, and it is not dominant, since it does not win against a strategy of Black that leads to $p$ and plays correctly after $m$.

\subsection{Sudoku}
\label{sec:sudoku}

Sudoku is a single-agent game in which, as in chess, a strategy with two mistakes can have no improving neighbor. Let $C$ be a set of cells and $D$ a set of digits, in the usual puzzle the $81$ squares of a $9 \times 9$ grid and $\{1, \dots, 9\}$. A filling is a function $g : C \to D$. I restrict attention to puzzles with a unique solution, so there is exactly one valid filling $g^{*}$, and the payoff of a filling is $w(g) = \one[g = g^{*}]$. There are no other players and no moves of nature, so there is a single contingency. I consider two games. In static Sudoku the agent submits a filling in one move, so the game has one situation, at which the actions are the fillings. The one-cell neighbors of a filling are the fillings that differ from it at one cell. In dynamic Sudoku the cells are filled in a fixed order, one cell per move, with the digits $D$ as the actions and no information revealed between moves other than the agent's own entries. By perfect recall, a situation is a cell together with the digits already entered. A strategy specifies a digit at every situation, and each strategy leads to one filling, so the two games have the same reduced normal-form game.

\begin{proposition}
\label{prop:sudoku}
(a) Static and dynamic Sudoku are SOSP, and their complexity in the measure of \citet{nagelsaitto} is $1$.

(b) Static Sudoku is not local improvement solvable under one-cell neighbors, and dynamic Sudoku is not local improvement solvable under one-situation neighbors.
\end{proposition}

The proof is immediate from the fact that every wrong filling has payoff $0$ and $g^{*}$ has payoff $1$. For (a), the minimum in (\ref{eq:SOD}) after the correct digit at any cell, or after $g^{*}$ in static Sudoku, is at least the maximum after a wrong one, since that maximum is $0$. In dynamic Sudoku no information about the contingency arrives between moves, so its essential frame is static Sudoku. Every $g \neq g^{*}$ is also evidently never better than $g^{*}$. For (b), a filling with two or more wrong cells and all of its one-cell neighbors pay $0$, so it is locally unimprovable and not payoff-equivalent to $g^{*}$. In dynamic Sudoku, take a strategy that enters wrong digits at a cell $c$ and at a later cell $c'$ when all earlier entries are correct, and correct digits everywhere else. Its filling has one wrong cell, correcting $c$ alone leads to the wrong digit at $c'$, and every other single change either leaves a wrong cell or is at a situation that is not reached. So it is locally unimprovable with payoff $0$, while the recommended strategy pays $1$.

Both classifications in (a) record that the participant needs no contingent reasoning about others and no reasoning about her own later moves. They do not record that the participant must find $g^{*}$, and that comparing neighboring fillings does not help her find it. The comparison does not help because a single change of a wrong filling leaves the payoff at $0$ unless every other cell is already correct. The difficulty of finding $g^{*}$ is the primitive complexity of \citet{bdk23}, the number of operations needed to find the best element of a menu, here the fillings.\footnote{The errors that \citet{enkeshubatt23} document in choices without strategic content seem to be of this kind.}

The proposition has two consequences. The first concerns dynamic implementation. Dynamic Sudoku is SOSP, so the obviousness concepts classify it as the simplest kind of extensive-form game, and yet filling the cells one at a time does not seem to help. The comparison at a cell has the same payoff on both sides unless the rest of the strategy completes $g^{*}$, so it carries no information about the solution until then. Obvious strategy-proofness therefore does not seem sufficient for dynamic implementation to help.

The second consequence concerns what the two approaches measure. Static Sudoku is a single decision with one contingency and no later move. It is SOSP, the strongest of the obviousness concepts, because these concepts hold trivially for any single choice with one contingency. It nevertheless seems difficult, for a reason that has nothing to do with contingent reasoning or with planning one's own later moves, the only sources of difficulty the obviousness concepts measure. A concept that classifies static Sudoku as no more difficult than choosing between two objects does not appear to capture all of the difficulty of recognizing the dominant strategy. Local improvement solvability captures part of it. With a single contingency, the comparison of two neighboring strategies has no contingent reasoning to reduce. Local improvement solvability is then the condition that every local maximum of the payoff is a global maximum, as Section~\ref{sec:lis} notes, and part (b) shows that neither version of Sudoku meets it. As in the example of Theorem~\ref{thm:lis}(b), the neighbor relation thus captures a source of simplicity beyond contingent reasoning, one that concerns the search through the strategy set.

\subsection{What the two approaches measure}
\label{sec:measure}

\begin{table}[!htbp]
\centering
\small
\renewcommand{\arraystretch}{1.1}
\setlength{\tabcolsep}{3pt}
\begin{tabular}{lcc}
\hline
 & Worst-best comparison concepts & Local improvement solvable \\
\hline
Pick-a-number game, $\lambda > K$ & not OSP, complexity $K + 1$ & yes, one-step neighbors \\
Chess, White has a winning strategy & OSP for White & no, one-situation neighbors \\
Sudoku, static & SOSP, complexity $1$ & no, one-cell neighbors \\
Sudoku, dynamic & SOSP, complexity $1$ & no, one-situation neighbors \\
\hline
\end{tabular}
\caption{Classification of the three games, Proposition~\ref{prop:additive}, Section~\ref{sec:chess}, and Proposition~\ref{prop:sudoku}.}
\label{tab:games}
\end{table}

Table~\ref{tab:games} collects the classifications of the three games. Both approaches describe ways to simplify the recognition of the dominant strategy, and they classify the three games differently. The worst-best comparison concepts, including the measure of \citet{nagelsaitto}, classify the extensive-form game by the contingent reasoning it requires. The pick-a-number game requires it if payoff levels are compared, and chess and Sudoku do not require it. Local improvement solvability classifies a game by the payoff differences between neighboring strategies, and asks whether their signs lead to the dominant strategy. This question arises even with a single contingency, as in Sudoku, where no contingent reasoning is needed. In the pick-a-number game every difference is $1$. In chess and Sudoku some strategies with two mistakes have no improving neighbor, although they are not dominant. The three games therefore suggest that the need for contingent reasoning is neither necessary nor sufficient for a mechanism to be difficult.

Local improvement solvability under the natural neighbor relation of each game thus accounts for part of what the worst-best comparison concepts miss. The explanation depends on that relation, since under the complete neighbor relation the pick-a-number game and Sudoku are local improvement solvable, being strategy-proof, and so is chess for White, who has a dominant strategy.

\section{Salience of the local comparison}
\label{sec:dynamic}

Section~\ref{sec:classical} gave the participant a path to the dominant strategy, in which each step compares two neighboring strategies and involves one case whose sign she determines from her own preference. The path helps her only if she notices it. A participant who is asked for her whole strategy at once may never compare a strategy with its neighbors. When a mechanism instead asks her at each step to choose between two neighbors, the question itself leads her along the path. Two strategy-proof mechanisms can therefore implement the same outcomes, be local improvement solvable for the same reason, and still differ in how easily the local comparison is noticed. The prediction of this section is that the one in which the comparison is more salient produces more dominant-strategy play.

The prediction compares two mechanisms. I define the relation between them formally and leave the salience of the local comparison informal. \begin{definition}[Implementing the outcomes]
\label{def:implements}
Let $\Ga$ and $\Ga'$, with outcome functions $f$ and $f'$, share players, type sets, and outcome space, and be strategy-proof under recommendations $\si$ and $\si'$. The mechanism $\Ga$ implements the outcomes of $\Ga'$ if there are maps $\phi_{i} : S_{i}' \to S_{i}$, the identity for nature, such that
\begin{equation}
f\big(s_{0}', \phi_{1}(s_{1}'), \dots, \phi_{n}(s_{n}')\big) \;=\; f'(s_{0}', s_{1}', \dots, s_{n}') \qquad \text{for every profile } s' \in S',
\label{eq:implements}
\end{equation}
and $\si_{i} = \phi_{i} \circ \si_{i}'$ for every player.
\end{definition}

In words, every strategy of $\Ga'$ is represented in $\Ga$ by a strategy that produces the same outcome against every profile of the others' strategies carried over by $\phi$, and the recommended strategies correspond. The two mechanisms therefore produce the same outcome at every type profile and every move of nature, $f(s_{0}, \si(\ta)) = f'(s_{0}, \si'(\ta))$, and implement the same choice rule in the sense of \citet{li17}.

The relation need not be symmetric. The sealed-bid second-price auction and the ascending clock auction of Section~\ref{sec:clock} implement each other's outcomes. The map $\phi$ sends a bid $b$ to the strategy that exits at $b$, and a map in the other direction sends each strategy of the ascending clock auction to its exit price. A strategy that stays in at every price below $K$ is sent to $K$. The dynamic serial dictatorship of Section~\ref{sec:dynamicmatching} implements the outcomes of serial dictatorship with list submission, with $\phi$ sending a list to the strategy that takes the top available object on the list. The converse fails, since the dynamic serial dictatorship has strategies whose outcomes no list reproduces. When $\Ga$ implements the outcomes of $\Ga'$, the map $\phi_{i}$ carries the neighbor relation of $\Ga'$ to the image $\phi_{i}(S_{i}')$. A run of the local improvement algorithm in $\Ga'$ is then a run in $\Ga$ along that image under one condition. Every contingency of $\Ga$ must give the participant, on that image, the payoffs of some contingency of $\Ga'$. The condition holds in the applications below, since there the others' play that affects her payoff cannot depend on her own.

Between two mechanisms of which one implements the outcomes of the other, the local comparison is more salient in the one that leads the participant to compare a strategy with its neighbors. A mechanism leads her to the comparison in one of two ways. It can ask her for a choice that contains a comparison between two neighboring strategies, as the ascending clock auction does at each price and the dynamic serial dictatorship at each menu. Or it can resolve her submitted strategy step by step in the order of the neighbor relation, as the passive clock of \citet{bsk22} does. A mechanism that asks for a bid, a vector of marginal bids, or a ranking in one move and resolves it at once does neither, however its rules are described. I use the ordering by salience only where the design makes it clear. The ordering is clear among the sealed-bid second-price auction, the passive clock, and the dynamic clock, and between list submission and a mechanism that asks for one object at a time. Among multi-unit auctions it is clear for the static Vickrey auction, the Ausubel auction, and the multi-account clock of Appendix~\ref{sec:accounts}. Among dynamic mechanisms, a question confined to two neighbors makes the comparison more salient than one that also offers strategies that are not neighbors.

Why is salience not formalized? It is a property of what the participant is shown and asked, not of the extensive-form mechanism. The worst-best comparison concepts take the extensive-form mechanism as their only input, and so does local improvement solvability once the neighbor relation is fixed. A formal measure of salience would need a model of what a participant attends to, and the evidence below does not require one.

The prediction does not deny that obvious strategy-proofness helps, since an OSP mechanism removes the need to reason about contingencies. The local comparison may be a second reason a participant finds the dominant strategy. I read the evidence below as evidence about the local comparison, with obvious strategy-proofness held fixed where the designs allow.

\subsection{The sealed-bid second-price auction and the ascending clock auction}
\label{sec:clock}

The ascending clock auction is defined as follows. There are $n \ge 2$ bidders and one object, and the agent's value is as in Section~\ref{sec:spa}. A price starts at $0$ and rises one grid step at a time through $\{0, \dots, K\}$. At each price $q$, every bidder who is still active chooses to stay in or to exit. A bidder who exits is inactive for the rest of the auction and observes nothing further, and an active bidder observes only the current price.\footnote{In the design of \citet{li17} and in some treatments of \citet{bsk22}, active bidders also see when rivals exit, information that is payoff-irrelevant under private values and is discussed below.} If exactly one bidder stays in at a price $q < K$, she wins the object and pays $q$. If all active bidders exit at the same price, or if two or more bidders are still active when the price reaches $K$, the object is not sold and no one pays. A strategy of the agent specifies, at every price at which she could still be active, whether to stay in or to exit. The bidder who plans to exit at price $b$, staying in at every lower price, obtains the object exactly when every rival exits before $b$. In that case she pays the price at which the last rival exits. She gets nothing when some rival exits at $b$ or later. Her payoff is therefore that of the bid $b$ in the sealed-bid second-price auction of Section~\ref{sec:spa}, including the tie rule, so the two auctions implement each other's outcomes.

Strategy-proofness alone therefore does not distinguish the two auctions. Yet \citet{li17} finds that the ascending clock auction produces more dominant-strategy play than the sealed-bid second-price auction and explains the difference by obvious strategy-proofness. The ascending clock auction is OSP. At each price, the worst outcome of the recommended action, staying in below the value and exiting at or above it, is at least as good as the best outcome of the other action. The sealed-bid second-price auction is not OSP. The ascending clock auction, however, also differs from the sealed-bid second-price auction in a second respect.

The ascending clock auction makes the local comparison salient. Exiting at $q$ and staying in at $q$ to decide again at $q + 1$ give the bidder different outcomes only if the last rival exits at $q$. In that case staying wins at price $q$ and exiting loses. The choice between exiting and staying in at $q$ is the comparison between the bids $q$ and $q + 1$ in Proposition~\ref{prop:spa}, restricted to the contingencies in which they differ. The ascending clock auction asks this comparison at every price in turn, and a bidder who answers each question according to her value follows a run of the local improvement algorithm of Figure~\ref{fig:flow}. The sealed-bid second-price auction asks for a single number and leads the bidder to no comparison.

Since the ascending clock auction is OSP and also makes the local comparison salient, Li's comparison changes both channels at once and cannot separate them. Displaying a clock while keeping the extensive-form mechanism of the sealed-bid second-price auction can separate the two channels.

\citet{bsk22} run such a design. They replicate the design of \citet{li17} and add intermediate treatments. In the passive clock treatment, subjects submit a sealed bid and then watch an ascending clock determine the outcome. In the dynamic clock treatment, subjects bid in an ascending clock auction without seeing when rivals exit. The passive clock has the same extensive-form mechanism as the sealed-bid second-price auction and is not OSP. Obvious strategy-proofness on its own therefore predicts no difference between the passive clock and the sealed-bid second-price auction, and the same holds for every worst-best comparison concept. The passive clock nevertheless makes the comparison between neighboring bids salient, since it resolves the bid price by price.

The passive clock removes about half of the difference in deviations from truthful bidding between the sealed-bid second-price auction and the ascending clock auction in which bidders see their rivals exit. Dynamic bidding, the feature that makes the ascending clock auction OSP, adds a further twelve percent, which is not statistically significant.

One finding in \citet{bsk22} lies outside both channels. Observing rivals' exits during a dynamic clock accounts for almost forty percent of the total effect, although this information is payoff-irrelevant under private values. The effect is significant only when the exits are seen while bidding. I attribute it to learning from the environment, a channel that neither obvious strategy-proofness nor the local improvement model of this paper contains.

This paper's explanation can account for the remaining findings, although Section~\ref{sec:limits} notes that it is not the only explanation that does. Under this explanation, the gain from the sealed-bid second-price auction to the passive clock comes from the local comparison alone. Most of the ascending clock auction's advantage that does not come from observing rivals' exits is already present in the passive clock. The explanation also predicts that the passive clock's effect lags that of the dynamic clock by about one round, which \citet{bsk22} report, as Appendix~\ref{app:further} explains.

Appendix~\ref{app:evidence} discusses further evidence consistent with this paper's explanation. The menu description of \citet{ght23} changes what a bidder in the sealed-bid second-price auction is told without changing what she is asked, and it leaves straightforward behavior unchanged.

The ascending clock auction also shows that the verdict of the criterion depends on the neighbor relation. It is local improvement solvable under one-step neighbors on exit prices, the neighbors its dynamic presentation makes salient, and not under one-situation neighbors, as Appendix~\ref{app:clockhamming} shows.

\subsection{The static Vickrey auction and the Ausubel auction}
\label{sec:ausubel}

The Ausubel auction of \citet{ausubel04} is a dynamic implementation of the multi-unit Vickrey rule of Section~\ref{sec:vickrey}. There are $Q$ identical units and $n$ bidders, and a price rises from $0$ one grid step at a time. At each price $q$ every bidder states a demand, a number of units no larger than the demand she stated at the previous price. A bidder observes only the current price and whether the auction is still open, the case that \citet{ausubel04} calls no bid information. She learns the units she has clinched when the auction ends. When the total demand of the other bidders at price $q$ is $D_{-i}(q)$, the number of units guaranteed to bidder $i$ is $\max\{0, Q - D_{-i}(q)\}$, capped at her own demand. Each unit that becomes guaranteed at price $q$ is clinched and paid for at $q$, and she keeps it even if her demand later falls below the number of units she has clinched. The auction ends at the first price at which total demand is at most $Q$. If total demand still exceeds $Q$ at price $K - 1$, the auction ends at $K$ without further clinching, and the units not clinched remain unsold.

A bidder who reduces her demand from $x$ to $x - 1$ at price $q$ withdraws her $x$-th marginal bid at $q$. A strategy that reduces demand from $\ell$ to $\ell - 1$ at price $a^{\ell}$ for each $\ell$, and never otherwise, has the payoff of the marginal bid vector $a$ in the static Vickrey auction. The Ausubel auction therefore implements the outcomes of the static Vickrey auction, with $\phi$ sending each vector to this strategy. As the definition of implementing the outcomes requires, sincere bidding is weakly dominant, since with no bid information every strategy has the payoff of some vector of marginal bids in the static Vickrey auction.\footnote{When bidders observe the others' bids, sincere bidding is an ex post equilibrium, and that version lies outside the class to which the prediction applies.}

The Ausubel auction makes the local comparison salient much as the ascending clock auction does, although at each price it asks for a number of units rather than a yes or a no. Reducing demand from $x$ to $x - 1$ units at price $q$, instead of holding it, changes the bidder's outcome only if the $x$-th unit would clinch at $q$. The choice between reducing and holding demand at $q$ is the comparison of Proposition~\ref{prop:vickrey} for unit $x$ at bid $q$. The Ausubel auction asks this comparison for each unit at each price, and a bidder who answers each according to the marginal value of the unit follows a run of the local improvement algorithm. The static Vickrey auction asks for a vector of marginal bids in one move, so the prediction is that the Ausubel auction produces more dominant-strategy play.\footnote{\citet{nagelsaitto} show that the Ausubel auction is the simplest implementation of the Vickrey outcome under their measure.} \citet{kagellevin09} compare the static Vickrey auction with two Ausubel auctions, one with and one without public dropout information. Bidding is significantly closer to sincere bidding in the version with dropout information than in either of the other two. In that version, however, sincere bidding is an ex post equilibrium rather than a dominant strategy. The gain resembles the effect of observed exits that Section~\ref{sec:clock} attributes to learning from the environment. A passive Ausubel auction, in which marginal bid vectors are submitted and then resolved price by price, would separate the salience of the local comparison from dynamic bidding. Appendix~\ref{sec:accounts} proposes a further design, a multi-account clock that asks the comparisons of Proposition~\ref{prop:vickrey} one unit at a time. Neither it nor the Ausubel auction is OSP for a bidder who demands two or more units, so comparing the two would test the salience of the local comparison with obvious strategy-proofness held fixed.

\subsection{Serial dictatorship and dynamic matching mechanisms}
\label{sec:dynamicmatching}

The dynamic serial dictatorship is the following extensive-form mechanism. The agents are called one at a time in the priority order, and each agent, when called, observes the set of objects with remaining capacity and picks one of them, which she receives. It implements the outcomes of serial dictatorship with list submission of Section~\ref{sec:matching}. In the random serial dictatorship treatments of \citet{li17}, the dynamic serial dictatorship produces the dominant-strategy outcome in 93 percent of games, against 64 percent for list submission. \citet{li17} explains the difference by obvious strategy-proofness. The dynamic serial dictatorship is OSP, and in fact SOSP in the sense of \citet{pyciatroyan23}, since the object the agent takes is the object she receives, while list submission is not OSP.

The dynamic serial dictatorship also makes the local comparison salient, and list submission does not. Suppose the agent is called and compares taking $a$ now with taking $b$ now, with the rest of her strategy fixed. In doing so she compares the list that has $a$ just before $b$ against the list that has $b$ just before $a$. These are two one-exchange neighbors, and by Proposition~\ref{prop:matching} they lead to different outcomes only where one gives her $a$ and the other gives her $b$. The most common error under list submission, the exchange of the two highest entries, is a single adjacent exchange away from the truthful list. This exchange seems to be the error of a participant who has not compared the two lists. As with the ascending clock auction, Li's comparison changes obvious strategy-proofness and the salience of the local comparison at once. The two explanations are separated only by mechanisms that change one without the other.

\citet{bohakimov23} provide such mechanisms. In a pick-an-object mechanism, every agent is first shown a menu of objects and picks one. A pick is tentative. An agent whose pick is rejected is shown a new menu, a strict subset of her previous one without the rejected pick, and picks again. Each agent observes only her own menus and picks, and when the procedure ends she receives the last object she picked. The recommended strategy picks the most preferred object on each menu. Since a pick can be rejected, the pick-an-object mechanisms discussed here are not OSP.\footnote{\citet{bohakimov23} show that an individually rational rule has a pick-an-object implementation in perfect ex post equilibrium if and only if the rule is strategy-proof and satisfies a condition they call monotonic discoverability.} Pick-an-object implementations exist for rules with no OSP implementation at all, including top trading cycles under cyclic priorities and deferred acceptance under general priorities.

For the pick-an-object serial dictatorship the recommended strategy is dominant, and the mechanism implements the outcomes of serial dictatorship with list submission, as Appendix~\ref{app:pao} shows. In the pick-an-object top trading cycles the menus of other agents can depend on an agent's picks, and I treat that mechanism, like iterative deferred acceptance in Appendix~\ref{app:ida}, as partial evidence.

The local comparison is more salient in a pick-an-object mechanism than under list submission. Each menu asks the agent for one object, and choosing $a$ over $b$ on a menu, with the rest of her strategy fixed, compares the one-exchange neighbors of Proposition~\ref{prop:matching}. The prediction is therefore that the pick-an-object serial dictatorship produces more dominant-strategy play than serial dictatorship with list submission. Obvious strategy-proofness on its own makes no such prediction, since the pick-an-object serial dictatorship is not OSP.

\citet{bohakimov23} find gains that vary in size and in the aspect of behavior they show up in. Under top trading cycles the rate of truthful play rises from 35 percent under list submission to 55 percent under pick-an-object with cyclic priorities, and from 56 to 67 percent with acyclic priorities. Both differences are significant. Under serial dictatorship the rate rises from 72 to 75 percent, which is not significant. But the objects that agents receive are significantly better under pick-an-object than under list submission, measured by the average rank of the assigned object in the agent's true preference. This suggests that the errors that remain under pick-an-object are more often ones that do not change the outcome. Under the obviously strategy-proof dynamic mechanisms, which have both channels, truthful play is 98 percent for serial dictatorship and 77 percent for top trading cycles with acyclic priorities. Yet the assigned objects are no better than under pick-an-object.

A dynamic mechanism asks for one choice at a time. In each dynamic mechanism of this section, the choice at a move contains a comparison between two neighbors under the neighbor relation of Section~\ref{sec:classical}. These mechanisms therefore make the local comparison salient, one comparison at a time. Making the comparison salient should help when the mechanism is local improvement solvable under those neighbors, so that each comparison has a sign that leads toward the dominant strategy. Dynamic play does not appear to be required for the effect of salience, since the passive clock makes the comparison salient without it. Obvious strategy-proofness is a separate channel. It does not appear to be required for the effect, as the passive clock and, in part, the pick-an-object serial dictatorship and iterative deferred acceptance suggest. Nor does obvious strategy-proofness seem sufficient for dynamic implementation to help, as dynamic Sudoku in Section~\ref{sec:sudoku} suggests.

\subsection{What the explanation does not account for}
\label{sec:limits}

Table~\ref{tab:evidence} in Appendix~\ref{app:evidence} collects the comparisons discussed in this section, and several of them fit this paper's explanation only in part or not at all. The passive clock of \citet{bsk22} shows subjects the clock after they have bid. Its effect could therefore come from learning from the feedback of watching their own bid resolved, and not only from the salience of the comparison between neighboring bids.\footnote{\citet{bsk22} describe the clock display as inducing a reasoning process that runs through the bids in ascending order, which is close to this paper's explanation. They also stress that subjects discover dominance by learning from feedback.} The lag of about one round fits both explanations, so the passive clock separates the salience of the local comparison from obvious strategy-proofness without separating salience from learning.

Other findings go against this paper's explanation or lie outside it. Observing rivals' exits in the ascending clock auction, and observing dropouts in the Ausubel auction of \citet{kagellevin09}, move bidding closer to the recommended strategy. Yet neither changes obvious strategy-proofness or the salience of the local comparison. Under the pick-an-object serial dictatorship, truthful play does not rise significantly, and the gain appears only in the objects that agents receive.

The rationing procedure of \citet{sbv} makes the local comparison salient and raises the rate of actions consistent with the dominant strategy. The obviously strategy-proof procedure for the uniform rule that they test asks each participant, step by step, whether to stop or to let her allocation move one more unit. It nevertheless produces the outcome of the uniform rule no more often than the direct mechanism, which the authors attribute to mistakes that accumulate over the several moves the procedure requires.\footnote{A mechanism with pre-play feedback, which is not OSP, performs better.}

Finally, the criterion is silent on the direction of errors. Under one-step neighbors the only locally unimprovable bids in the sealed-bid second-price auction are the truthful ones, so the criterion does not predict that subjects overbid, as \citet{kagel87} find, rather than underbid. This paper's explanation is a claim about which of two mechanisms produces more dominant-strategy play. It seems to fit most of the comparisons in Table~\ref{tab:evidence}, but it is not a model of the errors that remain.

\section{Design recommendations and conclusion}
\label{sec:design}

Strategy-proof mechanisms can be simple in a second way that the worst-best comparison concepts do not measure, since these concepts do not fully use the structure of the strategy set. In some mechanisms, every chain of moves to an improving neighbor ends at a dominant strategy. Local improvement solvability formalizes this property, and strategy-proofness is its weakest version. Many classical strategy-proof mechanisms are local improvement solvable under their natural neighbor relations, in a rich environment for the matching mechanisms, and each of their comparisons reduces to one case. In extensive-form mechanisms and games without duplicate dominant strategies, the property under one-situation neighbors is strict action-proofness, and it neither implies nor is implied by obvious strategy-proofness. Under each game's natural neighbor relation, the criterion accounts for part of what the worst-best comparison concepts miss in chess, Sudoku, and the pick-a-number game. The evidence, above all the passive clock, suggests that the local comparison helps participants when a mechanism makes it salient, even with obvious strategy-proofness held fixed.

Two design recommendations follow. The first is to choose, among mechanisms that implement the same outcomes, one that is local improvement solvable under its natural neighbor relation, the relation that the presentation of its strategy set defines. Without this property the comparison between neighbors can stall before reaching the dominant strategy, as it does in chess and Sudoku. Strong obvious strategy-proofness does not guarantee local improvement solvability under the natural neighbor relation, since static Sudoku is SOSP and is not local improvement solvable under one-cell neighbors.

The second recommendation is to make the local comparison salient, so that participants see the path when a mechanism has one. Two presentations do this, and the evidence for the first seems stronger.

First, where submissions must be simultaneous or sealed, resolve them step by step while the participant observes, as the passive clock of \citet{bsk22} does for the sealed-bid second-price auction. A passive presentation leaves the extensive-form mechanism, and therefore its incentive properties, unchanged.\footnote{A sealed-bid online auction, for example, can display the price rising after the bids are in, and a sealed-bid procurement auction can display it falling. A related static presentation can elicit willingness to pay by the Becker-DeGroot-Marschak procedure with a list that asks whether to buy at each price, rather than asking for a single number.} The evidence for the passive presentation separates the salience of the local comparison from obvious strategy-proofness, although not from learning from feedback.

Second, implement the rule dynamically where the setting permits. The ascending clock auction, the dynamic serial dictatorship, pick-an-object mechanisms, and iterative deferred acceptance each ask the participant for one choice at a time, and each choice contains a comparison between neighboring strategies. \citet{li17}, \citet{bsk22}, \citet{bohakimov20}, and \citet{bohakimov23} find corresponding gains, including for rules that have no obviously strategy-proof implementation.\footnote{Appendix~\ref{sec:accounts} proposes a design that would test the second recommendation in multi-unit auctions with obvious strategy-proofness held fixed.} The rationing procedure of \citet{sbv}, which also asks one comparison at a time and shows no gain in outcomes, is a reminder that a dynamic implementation does not guarantee better outcomes.

The criterion captures only one dimension of simplicity. It takes the neighbor relation as given and does not say which relation a participant actually uses. It records whether every run of local comparisons ends at the dominant strategy, not how hard each comparison is to evaluate, which is what the worst-best comparison concepts measure. And it can hold only in strategy-proof mechanisms. The first and the last of these limits lead to two of the open questions below.

Three questions are open. The first concerns mechanisms that are not strategy-proof, where the criterion cannot hold. In the strategy-proof mechanisms of Section~\ref{sec:classical}, the comparison between two neighbors has a sign that does not depend on a belief about the others. Iterative deferred acceptance already lies outside the class of strategy-proof mechanisms. Its recommended strategy is an equilibrium strategy rather than a dominant one, and yet asking for one choice at a time appears to bring a similar gain there. Do comparisons between neighbors lead a participant to a best response given her belief about the others? Which condition on the neighbor relation and the outcome function makes every run of the local improvement algorithm end at such a best response? And how does that condition relate to the equilibrium concepts under which mechanisms such as iterative deferred acceptance are analyzed?

The second question concerns design, and it stays within the strategy-proof rules. Local improvement solvability under a natural neighbor relation is stronger than strategy-proofness, so a rule implemented by a mechanism that satisfies it must be strategy-proof. Which strategy-proof rules have such an implementation? \citet{yangdworczak26} impose a requirement of this kind on menus. They call an ordered menu searchable if, for every type, every local maximum of utility over the menu is global, up to a condition on ties. They show that searchability restricts multidimensional screening. For a finite menu, the condition on local maxima is local improvement solvability for every type under one-step neighbors on the menu's order. This paper may therefore serve as a micro-foundation for their approach. It argues that comparing neighbors can simplify the search for the best strategy, and Section~\ref{sec:dynamic} suggests when participants make such comparisons. The question remains open for rules with several participants.

The third question concerns a model of salience. A model of what a description or a presentation must do to induce a given neighbor relation, in the spirit of the descriptions of \citet{ght23}, would make the prediction sharper.\footnote{The finding of \citet{bsk22} that observing rivals' exits accounts for a large share of the effect of the ascending clock auction lies outside such a model. Section~\ref{sec:clock} attributes the effect of observed exits to learning from the environment. Combining the comparison of neighbors with a model of learning is a further direction.} What a mechanism asks a participant to compare, and not only what the mechanism is, appears to shape whether she finds the dominant strategy.

\newpage
\appendix
\counterwithin*{table}{section}
\counterwithin*{figure}{section}
\renewcommand{\thetable}{\thesection.\arabic{table}}
\renewcommand{\thefigure}{\thesection.\arabic{figure}}

\section{Proofs}
\label{app:proofs}

\subsection{Proof of Proposition \ref{prop:coarse}}

Part (a). A locally unimprovable strategy $s_{i}$ under $N_{i}'$ is a strategy that no strategy in $N_{i}'(s_{i})$ improves. Since $N_{i}(s_{i}) \subseteq N_{i}'(s_{i})$, no strategy in $N_{i}(s_{i})$ improves it either, so it is locally unimprovable under $N_{i}$. Every locally unimprovable strategy under $N_{i}'$ is therefore a locally unimprovable strategy under $N_{i}$, all of which are payoff-equivalent by local improvement solvability under $N_{i}$. Hence $\Ga$ is local improvement solvable for $\ta_{i}$ under $N_{i}'$.

Part (b). Under the complete neighbor relation, a strategy is locally unimprovable if no strategy improves it. If $\Ga$ is local improvement solvable under the complete neighbor relation, then $\ta_{i}$ has a dominant strategy by Theorem~\ref{thm:lis}(a). Conversely, let $s_{i}^{*}$ be dominant and let $p$ be locally unimprovable. Then $s_{i}^{*}$ gives at least the same utility as $p$ in every contingency. If it gave a higher utility in some contingency, it would improve $p$, and $p$ would not be locally unimprovable. So $p$ is payoff-equivalent to $s_{i}^{*}$, all locally unimprovable strategies are payoff-equivalent, and $\Ga$ is local improvement solvable for $\ta_{i}$ under the complete neighbor relation. $\square$

\subsection{The measure of Nagel and Saitto}
\label{app:ns}

I state the measure of \citet{nagelsaitto} for a single strategic player $i$ in an extensive-form mechanism, the other players and nature being folded into the contingencies. For a situation $h \in H_{i}$ and a strategy $\rho \in S(h)_{i}$, let $Z(\rho, h) = \{u_{i}(\rho, s_{-i}; \ta_{i}) : s_{-i} \in S(h)_{-i}\}$ be the set of payoffs the player can receive when play reaches $h$ and she follows $\rho$ from $h$ on. For an action $a \in A(h)$, let $Z(a, h)$ be the union of $Z(\rho, h)$ over the strategies $\rho \in S(h)_{i}$ with $\rho_{h} = a$, the payoffs that can follow $a$ whatever she does later.

\begin{definition}[Evident reduction]
\label{def:ns}
At a situation $h$, action $a$ is evidently never better than $a'$ if $\max Z(a, h) \le \min Z(a', h)$. The evident reduction at $h$ is the set of actions that survive iterated removal, one action at a time, of an action evidently never better than some other surviving action.
\end{definition}

\begin{definition}[Essential frame and complexity]
\label{def:nscomplexity}
An essential frame merges two consecutive situations of the player whenever the information revealed between them excludes no payoff possible under her dominant continuation. An action at the merged situation specifies her actions at both. The complexity of an extensive-form mechanism is the largest number of actions in an evident reduction of its essential frame, over the players, their types, and their situations.
\end{definition}

The complexity is $1$ for the simplest mechanisms, and in the most complex static mechanisms it equals the number of strategies. \citet{nagelsaitto} define an essential frame directly, as a strategy-proof mechanism in which new relevant information is revealed to the agent between every two of her consecutive situations and whose tree is minimal. They define the complexity of any strategy-proof mechanism through an evident reduction of the player's continuation behaviors, which in an essential frame becomes the comparison of actions in Definition~\ref{def:ns}. They also show that every strategy-proof mechanism has an essential frame with the same outcomes and weakly lower complexity. The two definitions above simplify their construction.

\subsection{Proof of Proposition \ref{prop:vickrey}}

The residual prices are functions of the rivals' bids, weakly increasing in $\ell$, and do not depend on the agent's strategy. By the allocation rule, unit $\ell$ is won if and only if $a^{\ell} > p_{\ell}$. The strategy is weakly decreasing and the residual prices are weakly increasing, so the units the agent wins are units $1$ through $x$ for some $x$, and the payoff is (\ref{eq:unitsum}). Changing $a^{\ell}$ to $a^{\ell} + 1$ changes the term for unit $\ell$ only, and only in contingencies with $p_{\ell} = a^{\ell}$, where it changes the term from $0$ to $\ta^{\ell} - a^{\ell}$. Some profile of rival bids gives $p_{\ell} = a^{\ell}$, for example one in which the $Q$ highest rival bids all equal $a^{\ell}$. Three consequences follow, provided the changed strategy is weakly decreasing. The change from $a^{\ell}$ to $a^{\ell} + 1$ is a local improvement when $a^{\ell} < \ta^{\ell}$. The change from $a^{\ell}$ to $a^{\ell} - 1$ is a local improvement when $a^{\ell} - 1 > \ta^{\ell}$. A change between two bids in $T(\ta^{\ell})$, which occurs only when $\ta^{\ell}$ is a grid point and the bids are $\ta^{\ell}$ and $\ta^{\ell} + 1$, changes no payoff.

Write $\un{t}^{\ell}$ and $\ov{t}^{\ell}$ for the smallest and largest elements of $T(\ta^{\ell})$. Since $\ta^{\ell} > \ta^{\ell + 1}$, every element of $T(\ta^{\ell})$ is at least every element of $T(\ta^{\ell + 1})$. If $\ta^{\ell+1}$ is a grid point, then $\ov{t}^{\ell+1} = \ta^{\ell+1} + 1 \le \un{t}^{\ell}$, because $\un{t}^{\ell}$ is a grid point at or above $\ta^{\ell} > \ta^{\ell+1}$. Otherwise $\ov{t}^{\ell+1}$ is the least grid point above $\ta^{\ell+1}$, which is at most $\un{t}^{\ell}$.

Feasibility. Let $a$ be a weakly decreasing strategy that is not truthful. Suppose first that some unit has $a^{\ell} > \ov{t}^{\ell}$, and take the largest such $\ell$. Then $a^{\ell} - 1 \ge \ov{t}^{\ell} \ge \ta^{\ell}$, and $a^{\ell} - 1 > \ta^{\ell}$, because $a^{\ell} - 1 = \ta^{\ell}$ would put $a^{\ell}$ in $T(\ta^{\ell})$. If $\ell < L$, then $a^{\ell + 1} \le \ov{t}^{\ell + 1} \le \un{t}^{\ell} \le a^{\ell} - 1$. So lowering $a^{\ell}$ by one keeps the strategy weakly decreasing, and the change is a local improvement. Suppose instead that every unit has $a^{\ell} \le \ov{t}^{\ell}$. Then some unit has $a^{\ell} < \un{t}^{\ell}$, so $a^{\ell} < \ta^{\ell}$. Take the smallest such $\ell$. If $\ell > 1$, then $a^{\ell - 1} \ge \un{t}^{\ell - 1} \ge \ov{t}^{\ell} \ge \un{t}^{\ell} \ge a^{\ell} + 1$. So raising $a^{\ell}$ by one keeps the strategy weakly decreasing, and the change is a local improvement. Hence no strategy that is not truthful is locally unimprovable.

Every truthful strategy is locally unimprovable. Each of its neighbors either moves some $a^{\ell}$ to the other element of $T(\ta^{\ell})$, which changes no payoff, or moves $a^{\ell}$ outside $T(\ta^{\ell})$. By the first paragraph, the second kind of move lowers the payoff in the contingencies where $p_{\ell}$ equals the lower of the two bids and changes it nowhere else. All truthful strategies give the same utility in every contingency. Each unit's term depends on $a^{\ell}$ only through whether $a^{\ell} > p_{\ell}$, and the two elements of $T(\ta^{\ell})$, when there are two, differ in this only at $p_{\ell} = \ta^{\ell}$, where the term is $0$ either way. $\square$

\subsection{Proof of Proposition \ref{prop:matching}}

Part (a). Fix a type $\ta$ and a strategy $s_{i} \neq \ta$. Some adjacent pair of $s_{i}$ is ordered differently from $\ta$, since otherwise every adjacent pair of $s_{i}$ would be ordered as in $\ta$ and $s_{i}$ would equal $\ta$. Let $s_{i}'$ be the strategy obtained from $s_{i}$ by exchanging such a pair $o, o'$, where $\ta$ ranks $o$ above $o'$ and $s_{i}'$ lists $o$ first. By the adjacent-swap property applied to $s_{i}'$ and $s_{i}$, for every $s_{-i}$ either $f(s_{i}', s_{-i}) = f(s_{i}, s_{-i})$ or $f(s_{i}', s_{-i}) = o$ and $f(s_{i}, s_{-i}) = o'$. Hence $s_{i}'$ gives at least the same utility as $s_{i}$ in every contingency and a higher utility in every contingency where the pair is decisive. Richness provides such a contingency. So $s_{i}'$ locally improves $s_{i}$, and no strategy other than $\ta$ is locally unimprovable. The strategy $\ta$ is locally unimprovable. Each of its neighbors exchanges a pair into the order opposite to $\ta$ and, by the same argument with the roles exchanged, gives at most the same utility in every contingency. Every run therefore ends at $\ta$, and $\ta$ is dominant by Theorem~\ref{thm:lis}.

Part (b). Fix a contingency $s_{-i}$ and two strategies $s_{i}$ and $s_{i}'$ that differ by an adjacent exchange of $o$ and $o'$, with $s_{i}$ listing $o$ first. Let $x = f(s_{i}, s_{-i})$ and $y = f(s_{i}', s_{-i})$, and suppose $x \neq y$. The type whose true ranking is $s_{i}$ exists by the full domain. Strategy-proofness for that type says that she does at least as well with $s_{i}$ as with $s_{i}'$, so $s_{i}$ ranks $x$ above $y$. Strategy-proofness for the type whose true ranking is $s_{i}'$ says that $s_{i}'$ ranks $y$ above $x$. The two lists order every pair of objects in the same way except the pair $\{o, o'\}$, so $\{x, y\} = \{o, o'\}$, and since $s_{i}$ lists $o$ first, $x = o$ and $y = o'$. The only difference an adjacent exchange can make to the outcome is to exchange the two objects that were exchanged in the list. $\square$

\subsection{Richness in the three matching mechanisms}
\label{app:rich}

This subsection verifies the example of a rich environment given after Corollary~\ref{cor:matching}. Fix a list and an adjacent pair at positions $j$ and $j+1$. In serial dictatorship, the pair is decisive in a contingency in which the agent moves at position $j$ or later and her predecessors exhaust the capacities of exactly the $j - 1$ objects above the pair. In top trading cycles, the pair is decisive in a contingency in which two conditions hold. First, at each of the $j - 1$ objects above the pair, the agents with the highest priority, as many as its capacity, are other agents who rank that object first. Second, the agent has top priority at both objects of the pair. Until its capacity is exhausted, each object above the pair points to one of these agents, who points back to it. The agent points to such an object while one has capacity, so she is in no cycle until all of them are exhausted. Both objects of the pair point to her while she remains, so neither loses capacity before she is assigned, and she then receives the one her list ranks first. In deferred acceptance, the pair is decisive in a contingency in which two conditions hold. First, the $j - 1$ schools above the pair fill their capacities with students who rank them first and have higher priority there than the agent. Second, the agent has top priority at both schools of the pair. From the first round on, each school above the pair is then full with students of higher priority than the agent, so it rejects her. Neither school of the pair rejects her, so she is assigned the one her list ranks first. Suppose the agent's priority is part of the contingency, every profile of priorities and others' strategies is possible, and there are enough other agents to exhaust the capacities of the objects above the pair. Then each of these contingencies exists. $\square$

When the agent knows that some adjacent pair on her list cannot affect her outcome, the two strategies that exchange the pair are payoff-equivalent, and the exchange is not a local improvement. Suppose, for example, that the agent knows that she is second in a serial dictatorship over four objects, one unit each, and her true ranking is $abcd$. The list $acdb$ has no local improvement. Exchanging $d$ and $b$ changes no outcome, and the other two adjacent exchanges lower her payoff. Yet the list gives her $c$ instead of $b$ when the first dictator takes $a$. Richness rules out such a list.

\subsection{Dominance at a situation}
\label{app:dominance}

This subsection proves the equivalence of (\ref{eq:SPdef}) and (\ref{eq:D}) in Definition~\ref{def:dominant}. Condition (\ref{eq:SPdef}) implies (\ref{eq:D}), since (\ref{eq:D}) restricts (\ref{eq:SPdef}) to some strategies and contingencies. For the converse, suppose a strategy reaches a different terminal node from $s_{i}^{*}$ in some contingency. Then the two strategies' action sequences in that contingency first differ at some situation $h \in H_{i}$. Both strategies are therefore in $S(h)_{i}$, the contingency reaches $h$, and the strategies specify different actions at $h$. Hence (\ref{eq:D}) at $h$ gives the inequality in (\ref{eq:SPdef}) for that contingency. In a contingency in which the two strategies reach the same terminal node, they give the same utility. $\square$

\subsection{Proof of Theorem \ref{prop:hamming}}

Two strategies that specify the same action at every situation both of them reach lead to the same play path in every contingency, and are therefore payoff-equivalent.

From strict action dominance to local improvement solvability. Let $s_{i}^{*}$ be strictly action dominant. It is dominant by Proposition~\ref{prop:taxonomy}(a), hence locally unimprovable. Let $p$ be locally unimprovable and suppose $p$ and $s_{i}^{*}$ both reach some $h$ at which they differ. Then $p \in S(h)_{i}$, and by strict action dominance $p[h \leftarrow s^{*}_{h}]$ locally improves $p$, a contradiction. So every locally unimprovable strategy agrees with $s_{i}^{*}$ wherever both reach and is payoff-equivalent to it. A strictly action dominant strategy also rules out duplicates. Suppose a dominant strategy $p$ differed from $s_{i}^{*}$ at a situation $h$ that both reach. Strict action dominance at $h$ with $r = p$ would then make $p[h \leftarrow s^{*}_{h}]$ strictly better than $p$ in some contingency, which contradicts the dominance of $p$. So every dominant strategy agrees with $s_{i}^{*}$ wherever both reach. If two dominant strategies differed at a situation $g$ that both reach, $s_{i}^{*}$ could not reach $g$, since it would agree with both there. It would then differ from both at the situation where it leaves the path to $g$, which all three reach, a contradiction.

From local improvement solvability to strict action dominance. Let $\Ga$ be local improvement solvable under one-situation neighbors with no duplicate dominant strategies for $\ta_{i}$, and let $s_{i}^{*}$ be a dominant strategy, which exists by Lemma~\ref{lem:terminate} and Theorem~\ref{thm:lis}(a). Fix a situation $h$ that $s_{i}^{*}$ reaches, a strategy $r \in S(h)_{i}$, and an action $a' \neq s^{*}_{h}$, and suppose strict action dominance fails for this situation, strategy, and action. Then either $u_{i}(r[h \leftarrow a'], s_{-i}^{0}; \ta_{i}) > u_{i}(r[h \leftarrow s^{*}_{h}], s_{-i}^{0}; \ta_{i})$ for some $s_{-i}^{0}$ that reaches $h$, or $r[h \leftarrow a']$ and $r[h \leftarrow s^{*}_{h}]$ are payoff-equivalent. Let $q = r[h \leftarrow a']$, and call the situations $h'$ of player $i$ with $S(h')_{h} = \{s^{*}_{h}\}$ the recommended subtree. Run the local improvement algorithm from $q$ with changes restricted to situations outside the recommended subtree, and let $p$ be the strategy at which it stops. Along the run the utility weakly increases in every contingency, so $p$ does at least as well as $q$ everywhere.

The run never changes the action at $h$ to $s^{*}_{h}$. Such a change, from a strategy $p'$ to $p'[h \leftarrow s^{*}_{h}]$, alters no outcome if $p'$ does not reach $h$. If $p' \in S(h)_{i}$, it alters play only in contingencies that reach $h$. In those contingencies, $p'[h \leftarrow s^{*}_{h}]$ plays as $r[h \leftarrow s^{*}_{h}]$, because the actions in the recommended subtree are still those of $r$. In the first case the changed strategy is worse than $q$, and hence than $p'$, at $s_{-i}^{0}$. In the second case it gives the utility of $q$ in every contingency that reaches $h$ and that of $p'$ in every other, hence at most the utility of $p'$ everywhere. So the change is not a local improvement. Hence $p_{h} \neq s^{*}_{h}$, no situation in the recommended subtree is reached under $p$, and the changes the run did not consider alter no outcome. So $p$ is locally unimprovable, hence dominant by Theorem~\ref{thm:lis}. Since there are no duplicates, $p$ agrees with $s_{i}^{*}$ at every situation both reach. If $p$ reaches $h$, this agreement fails at $h$. If not, then at the earliest situation $g$ preceding $h$ at which $p$ leaves the path to $h$, both strategies reach $g$ and differ there. Either way there is a contradiction, so $s_{i}^{*}$ is strictly action dominant. $\square$

\subsection{Proof of Proposition \ref{prop:taxonomy}}

Part (a). Fix $s_{i}^{*}$, a situation $h$ it reaches, and $a' \neq s^{*}_{h}$, with minima and maxima ranging as in Definitions~\ref{def:od} and~\ref{def:sod}. $(\ref{eq:SOD}) \Rightarrow (\ref{eq:OD})$ holds because $s_{i}^{*}$ is a continuation of $s^{*}_{h}$ and every deviation $s_{i}'$ at $h$ with $s'_{h} = a'$ is a continuation of $a'$. Both this implication and $(\ref{eq:OD}) \Rightarrow (\ref{eq:D})$ use the fact that a minimum is at most and a maximum at least each element of its set. $(\ref{eq:SOD}) \Rightarrow (\ref{eq:AD})$ holds because $r[h \leftarrow s^{*}_{h}]$ and $r[h \leftarrow a']$ are continuations of $s^{*}_{h}$ and $a'$. Hence, in every contingency that reaches $h$, the first pays at least the minimum in (\ref{eq:SOD}) and the second at most the maximum, while in other contingencies they coincide.

$(\ref{eq:AD}) \Rightarrow (\ref{eq:D})$. Let $s_{i}'$ deviate from $s_{i}^{*}$ at $h$, so that $s'_{h} = a'$. Define $r$ by $r_{h''} = s'_{h''}$ at every situation $h''$ that follows $h$ through $a'$ and $r_{h''} = s^{*}_{h''}$ at every other situation. Then $r \in S(h)_{i}$, since $r$ agrees with $s_{i}^{*}$ at the situations preceding $h$. Fix $s_{-i} \in S(h)_{-i}$. The play of $(r[h \leftarrow s^{*}_{h}], s_{-i})$ passes through $h$, and after $h$ it meets only situations of player $i$ that follow $h$ through $s^{*}_{h}$. At all of these, and at the situations preceding $h$, $r$ agrees with $s_{i}^{*}$, so the play reaches the terminal node of $(s_{i}^{*}, s_{-i})$. In the same way $(r[h \leftarrow a'], s_{-i})$ reaches the terminal node of $(s_{i}', s_{-i})$, since $s_{i}'$ agrees with $s_{i}^{*}$ before $h$. Condition (\ref{eq:AD}) at $h$ therefore gives $u_{i}(s_{i}^{*}, s_{-i}; \ta_{i}) \ge u_{i}(s_{i}', s_{-i}; \ta_{i})$, which is (\ref{eq:D}) at $h$. Since $h$ and $s_{i}'$ were arbitrary, (\ref{eq:D}) holds at every situation that $s_{i}^{*}$ reaches, and by Appendix~\ref{app:dominance}, $s_{i}^{*}$ is dominant.

Parts (b) and (c). Section~\ref{sec:taxonomy} shows that the game of Figure~\ref{fig:twomove} is OSP and not AP and that the sealed-bid second-price auction is AP and not OSP. The other two examples are the games of Figure~\ref{fig:examples}. In each tree the other player $j$ moves first and chooses the contingency, $c_{1}$ or $c_{2}$. The agent $i$ does not observe this choice, so each of her situations contains two nodes, joined by a dashed line, and each terminal node shows her payoff only. Strategies are written as in Section~\ref{sec:hamming}, and a pair such as $(10, 3)$ lists the agent's payoffs at $c_{1}$ and $c_{2}$.

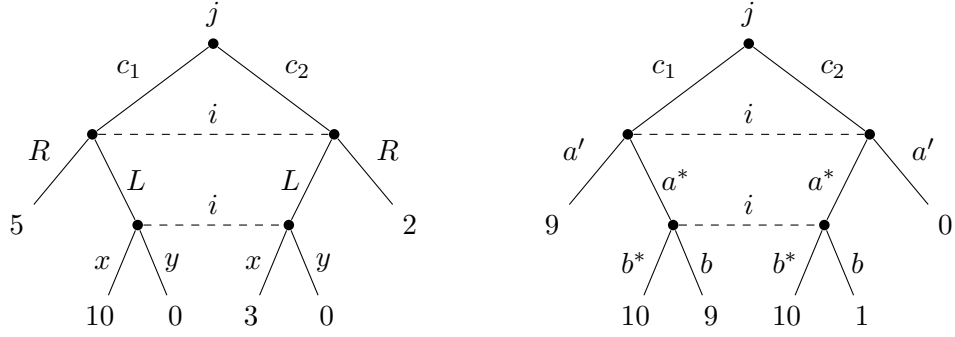
\begin{figure}[t]
\centering
\begin{tikzpicture}[x=1cm, y=1cm, every node/.style={font=\small}, dn/.style={circle, fill, inner sep=1.4pt}]
\node[dn, label=above:{$j$}] (r) at (0,0) {};
\node[dn] (A1) at (-1.6,-1.2) {};
\node[dn] (A2) at (1.6,-1.2) {};
\draw (r) -- node[above left] {$c_{1}$} (A1);
\draw (r) -- node[above right] {$c_{2}$} (A2);
\draw[dashed] (A1) -- node[above] {$i$} (A2);
\node (R1) at (-2.6,-2.4) {$5$};
\node[dn] (B1) at (-1.0,-2.4) {};
\node[dn] (B2) at (1.0,-2.4) {};
\node (R2) at (2.6,-2.4) {$2$};
\draw (A1) -- node[above left] {$R$} (R1);
\draw (A1) -- node[right] {$L$} (B1);
\draw (A2) -- node[left] {$L$} (B2);
\draw (A2) -- node[above right] {$R$} (R2);
\draw[dashed] (B1) -- node[above] {$i$} (B2);
\node (x1) at (-1.5,-3.6) {$10$};
\node (y1) at (-0.5,-3.6) {$0$};
\node (x2) at (0.5,-3.6) {$3$};
\node (y2) at (1.5,-3.6) {$0$};
\draw (B1) -- node[left] {$x$} (x1);
\draw (B1) -- node[right] {$y$} (y1);
\draw (B2) -- node[left] {$x$} (x2);
\draw (B2) -- node[right] {$y$} (y2);
\end{tikzpicture}
\hspace{1.2cm}
\begin{tikzpicture}[x=1cm, y=1cm, every node/.style={font=\small}, dn/.style={circle, fill, inner sep=1.4pt}]
\node[dn, label=above:{$j$}] (r) at (0,0) {};
\node[dn] (A1) at (-1.6,-1.2) {};
\node[dn] (A2) at (1.6,-1.2) {};
\draw (r) -- node[above left] {$c_{1}$} (A1);
\draw (r) -- node[above right] {$c_{2}$} (A2);
\draw[dashed] (A1) -- node[above] {$i$} (A2);
\node (R1) at (-2.6,-2.4) {$9$};
\node[dn] (B1) at (-1.0,-2.4) {};
\node[dn] (B2) at (1.0,-2.4) {};
\node (R2) at (2.6,-2.4) {$0$};
\draw (A1) -- node[above left] {$a'$} (R1);
\draw (A1) -- node[right] {$a^{*}$} (B1);
\draw (A2) -- node[left] {$a^{*}$} (B2);
\draw (A2) -- node[above right] {$a'$} (R2);
\draw[dashed] (B1) -- node[above] {$i$} (B2);
\node (x1) at (-1.5,-3.6) {$10$};
\node (y1) at (-0.5,-3.6) {$9$};
\node (x2) at (0.5,-3.6) {$10$};
\node (y2) at (1.5,-3.6) {$1$};
\draw (B1) -- node[left] {$b^{*}$} (x1);
\draw (B1) -- node[right] {$b$} (y1);
\draw (B2) -- node[left] {$b^{*}$} (x2);
\draw (B2) -- node[right] {$b$} (y2);
\end{tikzpicture}
\caption{The games of Proposition~\ref{prop:taxonomy}(b), left, and Proposition~\ref{prop:taxonomy}(c), right.}
\label{fig:examples}
\end{figure}

The game on the left of Figure~\ref{fig:examples} is SP and neither OSP nor AP. The strategy $(L, x)$ pays $(10, 3)$ and is dominant. It is not obviously dominant, since its worst payoff, $3$, is below the best payoff of $R$, $5$. It is not action dominant at the first situation, since $(L, y)$ pays $(0, 0)$ and $(R, y)$ pays $(5, 2)$.\footnote{Under one-situation neighbors the game is not local improvement solvable either. The strategy $(R, y)$ has no local improvement, since changing $R$ to $L$ gives $(0, 0)$ and changing $y$ changes nothing at an unreached situation.}

The game on the right of Figure~\ref{fig:examples} proves (c). The recommended strategy $(a^{*}, b^{*})$ pays $(10, 10)$, the best payoff in each contingency, so it is obviously dominant. Changing either of its actions weakly lowers the payoff in each contingency, whatever the other action is, so the recommended strategy is action dominant. It is not strongly obviously dominant, since the worst payoff after $a^{*}$, over the agent's continuations and the contingencies, is $1$, below the best payoff after $a'$, $9$. Strong obvious dominance compares the minimum over own continuations and contingencies jointly against the maximum over both jointly. Action dominance fixes the contingency, obvious dominance fixes the continuation, and neither compares a bad continuation in one contingency against the deviation in another. $\square$

\subsection{Proof of Proposition \ref{prop:additive}(c)}

Step one. Let $\phi$ and $\si$ be as in Definition~\ref{def:implements}, so that $\si_{i} = \phi_{i}(K)$ for each player. Let $h$ be the first situation of player $i$ on the play of $(\phi_{i}(a), \phi_{j}(b))$ at which $\phi_{i}(a)$ and $\phi_{i}(a')$ choose different actions, where $a < a'$ without loss of generality. It exists, since the profiles $(\phi_{i}(a), \phi_{j}(b))$ and $(\phi_{i}(a'), \phi_{j}(b))$ give the outcomes $(a, b)$ and $(a', b)$ and differ only in the strategy of $i$, so their plays first differ at a situation of $i$. The two strategies agree before $h$, so the play of $(\phi_{i}(a'), \phi_{j}(b))$ also reaches $h$. Suppose that some $\phi_{j}(b'')$ does not reach $h$. For $c \in \{a, a'\}$, let $\tau_{c}$ agree with $\phi_{i}(c)$ at $h$ and at every situation that follows $h$, and with the other of the two strategies elsewhere. Against $\phi_{j}(b)$ the strategy $\tau_{c}$ plays as $\phi_{i}(c)$. Against $\phi_{j}(b'')$ the plays of both strategies avoid $h$ and every situation that follows it, so $\tau_{c}$ plays as the other strategy. If $b'' = K$, then against $\tau_{a'}$ the strategy $\phi_{j}(b)$ gives $j$ the payoff $b + \lambda a' \ge \lambda a'$, and $\phi_{j}(K)$ gives her $K + \lambda a$. If $b = K$, then against $\tau_{a}$ the strategy $\phi_{j}(b'')$ gives $j$ at least $\lambda a'$, and $\phi_{j}(K)$ gives her $K + \lambda a$. Since $\lambda a' - K - \lambda a \ge \lambda - K > 0$, in either case $\phi_{j}(K)$ is not dominant, which contradicts strategy-proofness under $\si$. If neither $b$ nor $b''$ is $K$, the first case applies to $K$ in place of $b''$ if $\phi_{j}(K)$ does not reach $h$, and the second to $K$ in place of $b$ otherwise. Hence every $\phi_{j}(b')$ reaches $h$.

Step two. Fix $m < K$. Against $\phi_{j}(0)$ the strategies $\phi_{i}(K)$ and $\phi_{i}(m)$ give the outcomes $(K, 0)$ and $(m, 0)$. Let $h_{0}$ be the first situation of player $i$ on the play of $(\phi_{i}(K), \phi_{j}(0))$ at which the two strategies differ. By step one, $\phi_{j}(K)$ reaches $h_{0}$. Over the contingencies that reach $h_{0}$, the minimum payoff of $\phi_{i}(K)$ is at most $K$, at $\phi_{j}(0)$, and the maximum payoff of $\phi_{i}(m)$ is at least $m + \lambda K > K$, at $\phi_{j}(K)$. So (\ref{eq:OD}) fails at $h_{0}$. Every dominant strategy of $i$ pays what $\phi_{i}(K)$ pays in every contingency, $K + \lambda b$ against $\phi_{j}(b)$, and so gives the outcome $(K, b)$, since $\lambda > K$. Replacing $\phi_{i}(K)$ by it preserves Definition~\ref{def:implements}, so the same argument applies to it. Hence no extensive-form game that implements the outcomes of the static game is OSP.

Step three. Now suppose that, for every $a$, the strategy $\phi_{i}(a)$ gives $i$ the number $a$ against every strategy of $j$. Let $h$ be a situation of $i$ that $\phi_{i}(x)$ and $\phi_{i}(y)$ reach, with $x < y$, and suppose that some strategy of $j$ does not reach $h$. Then there are two strategies of $j$, one of them $\phi_{j}(K)$, of which exactly one reaches $h$. On the play of $\phi_{i}(x)$ against the one that reaches $h$, let $g$ be the first situation of $i$ that the play of $\phi_{i}(x)$ against the other does not reach. The strategies $\phi_{i}(x)$ and $\phi_{i}(y)$ agree before $g$, and against every strategy of $j$ they give different numbers. So the play of $\phi_{i}(x)$ against the other strategy reaches a situation $g' \neq g$ of $i$ after the same actions of $i$, and both strategies reach $g$ and $g'$. No play reaches both $g$ and $g'$, as neither precedes the other, so neither strategy of $j$ reaches both, and each reaches one of them. Let $\tau$ agree with $\phi_{i}(y)$ at the one of $g$ and $g'$ that $\phi_{j}(K)$ does not reach and at every situation that follows it, and with $\phi_{i}(x)$ elsewhere. Against the strategy of $j$ that reaches it, $\tau$ gives $i$ the number $y$, so that strategy of $j$ receives at least $\lambda y$. Against $\phi_{j}(K)$, the strategy $\tau$ gives $i$ the number $x$, so $\phi_{j}(K)$ receives at most $K + \lambda x$. Since $\lambda y - K - \lambda x \ge \lambda - K > 0$, this contradicts strategy-proofness. Hence every strategy of $j$ reaches every situation of $i$ that two of the strategies $\phi_{i}(a)$ reach, and nothing about $j$ is revealed between two consecutive such situations. In the construction of \citet{nagelsaitto}, two consecutive situations of player $i$ are merged unless the information revealed between them excludes some payoff possible under her dominant continuation. Nothing is excluded here, whether or not $\phi_{i}(K)$ reaches the two situations. Her first situation is such a situation, as all her strategies reach it, so every play reaches it. Every situation of $i$ before such a situation is again one. So the essential frame merges all these situations into a single situation. Against $\phi_{j}(0)$, every strategy of $i$ pays at most $K$, the payoff of the dominant strategy $\phi_{i}(K)$. So every action at the merged situation has a minimum payoff of at most $K$. The strategies $\phi_{i}(a)$ take $K + 1$ different actions there, since any two of them first differ at a situation that both reach. The action of $\phi_{i}(a)$ can be followed by $a + \lambda K > K$, against $\phi_{j}(K)$. Hence none of these actions is evidently never better than another action, all of them are in the evident reduction, and the complexity is at least $K + 1$. $\square$

\subsection{The chess argument}
\label{app:chess}

The strategy of Section~\ref{sec:chess} makes its other moves so that no other change at one position improves it. The winning strategy is chosen so that, whenever a position it reaches offers a move after which every continuation of White wins against all play of Black, it makes such a move. No position on the path to $p$, $p$ included, then offers such a move. At a position before $p$ the winning strategy would make it, and every continuation would then win at $p$, $m$ included. At $p$ the move would be $w$, and the move at $q$ could not be losing. At the other positions that the winning strategy reaches, the strategy plays as the winning strategy. After every other move that the winning strategy does not make, the strategy plays so that Black can prevent a win, where this is possible. After $m$ it plays a continuation that wins against some error of Black and that no change at one position improves. Such a continuation exists by Lemma~\ref{lem:terminate}, since a run of the local improvement algorithm after $m$, started from a continuation that wins against some error of Black, never lowers a payoff. A change after $m$ is then not an improvement. Nor is a change at any other reached position not on the path to $p$, since the strategy plays as the winning strategy there and wins against all play of Black through it. A change to a new move at $p$ or at a position before $p$ is not an improvement either. Against a strategy of Black that leads to $p$, errs after $m$ so that the original strategy wins, and prevents a win after the new move, the changed strategy does not win. $\square$

\section{Additional evidence}
\label{app:evidence}

This appendix collects the evidence of Section~\ref{sec:dynamic} in one table and discusses three groups of studies that the section mentions only briefly or not at all. It also checks richness in Li's random serial dictatorship design.

\subsection{Summary of the evidence}
\label{sec:evsummary}

Table~\ref{tab:evidence} collects the experimental comparisons between two mechanisms that produce the same outcomes, or between two descriptions of one mechanism. Each row changes a baseline into a treatment. The second column says whether the change makes the mechanism OSP, and the third says whether it makes the local comparison more salient. Rows marked with a dagger involve a recommended strategy that is an equilibrium strategy rather than a dominant strategy, so they are partial evidence. Section~\ref{sec:dynamic} and this appendix discuss the rows in turn.

The rows can be grouped by the channels they change. In the rows where only the salience of the local comparison changes, the gain is large for the passive clock. It is also present for the pick-an-object top trading cycles and iterative deferred acceptance, although these two rows are partial evidence. For the pick-an-object serial dictatorship, the gain appears in the objects received and not in truthful play. Where only obvious strategy-proofness changes, truthful play rises by 23 percentage points for serial dictatorship and by 10 for top trading cycles, although the objects received are no better. The gain is not significant for the dynamic clock. The ascending clock auction and the dynamic serial dictatorship change both channels and show the largest gains. The rationing procedure of \citet{sbv} also changes both channels, however, and shows no gain in outcomes, although actions consistent with the dominant strategy become more frequent. Section~\ref{sec:limits} takes up this result together with two effects that change neither channel.

\begin{table}[!htbp]
\centering
\footnotesize
\renewcommand{\arraystretch}{1.1}
\setlength{\tabcolsep}{4pt}
\begin{tabular}{p{4.0cm}ccp{3.9cm}p{2.9cm}}
\hline
\rule{0pt}{4.5ex}Baseline $\to$ treatment & \shortstack{Becomes\\OSP} & \shortstack{Comparison\\more salient} & Finding & Source \tabularnewline
\hline
\multicolumn{5}{l}{\emph{Auctions}} \\
\raggedright Sealed-bid $\to$ ascending clock & yes & yes & \raggedright more dominant-strategy play & \raggedright \citet{li17} \tabularnewline
\raggedright Sealed-bid $\to$ passive clock & no & yes & \raggedright about half of the gap to the clock closed & \raggedright \citet{bsk22} \tabularnewline
\raggedright Passive clock $\to$ dynamic clock & yes & no & \raggedright a further 12 percent, not significant & \raggedright \citet{bsk22} \tabularnewline
\raggedright Dynamic clock $\to$ dynamic clock showing rivals' exits & no & no & \raggedright almost 40 percent of the total effect & \raggedright \citet{bsk22} \tabularnewline
\raggedright Traditional $\to$ menu description of the sealed-bid second-price auction & no & no & \raggedright no difference & \raggedright \citet{ght23} \tabularnewline
\raggedright Ausubel without $\to$ with dropout information$^{\dagger}$ & no & no & \raggedright closer to sincere bidding than static Vickrey or Ausubel without dropout information & \raggedright \citet{kagellevin09} \tabularnewline
\multicolumn{5}{l}{\emph{Matching}} \\
\raggedright List submission $\to$ dynamic serial dictatorship & yes & yes & \raggedright 64 $\to$ 93 percent & \raggedright \citet{li17} \tabularnewline
\raggedright List submission $\to$ pick-an-object, serial dictatorship & no & yes & \raggedright 72 $\to$ 75 percent, not significant, better objects received & \raggedright \citet{bohakimov23} \tabularnewline
\raggedright Pick-an-object $\to$ OSP, serial dictatorship & yes & no & \raggedright 75 $\to$ 98 percent, objects received no better & \raggedright \citet{bohakimov23} \tabularnewline
\raggedright List submission $\to$ pick-an-object, top trading cycles$^{\dagger}$ & no & yes & \raggedright 35 $\to$ 55 percent (cyclic), 56 $\to$ 67 percent (acyclic) & \raggedright \citet{bohakimov23} \tabularnewline
\raggedright Pick-an-object $\to$ OSP, top trading cycles$^{\dagger}$ & yes & no & \raggedright 67 $\to$ 77 percent (acyclic), objects received no better & \raggedright \citet{bohakimov23} \tabularnewline
\raggedright List submission $\to$ iterative deferred acceptance$^{\dagger}$ & no & yes & \raggedright more straightforward play, more stable outcomes & \raggedright \citet{bohakimov20} \tabularnewline
\multicolumn{5}{l}{\emph{Rationing, uniform rule}} \\
\raggedright Direct $\to$ OSP procedure & yes & yes & \raggedright more actions consistent with the dominant strategy, uniform-rule outcome no more often & \raggedright \citet{sbv} \tabularnewline
\raggedright Direct $\to$ pre-play feedback & no & unclear & \raggedright more peak strategies and uniform-rule outcomes & \raggedright \citet{sbv} \tabularnewline
\hline
\end{tabular}
\caption{Experimental evidence on the two channels, Section~\ref{sec:dynamic}.}
\label{tab:evidence}
\end{table}

\subsection{Descriptions of the sealed-bid second-price auction}
\label{app:ght}

\citet{ght23} change what a bidder in the sealed-bid second-price auction is told without changing what she is asked. Their menu description tells the bidder that the rivals' bids fix a price $p$ and that her own bid determines only whether she buys at that price. This description exposes strategy-proofness without requiring the bidder to reason through the rules. With five bidders, they find no difference in straightforward behavior between the menu description and the traditional description, although the same kind of description raises straightforward behavior in a median-voting mechanism.

This paper's explanation is consistent with the auction result. The menu description asks for a single number, as the traditional description does. It does display a two-element choice, buying at $p$ or not, and this is the decisive case of Proposition~\ref{prop:spa} stated for a given price. But this choice is a comparison between two outcomes at a price the bidder does not know, not a comparison between two of her own bids. The menu description also does not tell her which bid to move to. Neither description makes the comparison between neighboring bids salient, so the two descriptions leave the bidder with the same task and should perform alike. A null result can have other causes, so the finding supports the explanation less strongly than the passive-clock results do. What seems to separate the passive and dynamic clocks from both descriptions is whether the bidder is led to compare a bid with the next bid up, not what she is told about the rules.

\subsection{Further evidence on the ascending clock auction}
\label{app:further}

Further evidence on the ascending clock auction is consistent with this paper's explanation. \citet{harstad00} finds that experience with the ascending clock auction partly carries over to later play in the sealed-bid second-price auction. The carryover is expected if what subjects learn from the ascending clock auction is a way of evaluating bids rather than a property of that auction. \citet{harstad00} reads part of it this way, as considering prices one by one, but most of it as an awareness that overbidding leads to losses. \citet{schotter94} find, in a different game, that presenting a static game as a sequence of moves increases dominant-strategy play. \citet{georganas17} find that changes in payoffs away from the dominant strategy change behavior sharply in the sealed-bid second-price auction. By contrast, \citet{mcgeelevin19} find that similar changes alter bidding little in the ascending clock auction. The difference is expected if subjects in the sealed-bid second-price auction evaluate payoffs directly while subjects in the ascending clock auction answer the question the mechanism asks. The Becker-DeGroot-Marschak procedure is a one-bidder second-price auction against a random price. Its price-list version, a table asking whether to sell or keep at each listed price, asks for one comparison at each price instead of a single number. Under the single-number version of the same procedure, \citet{casonplott14} find that on their first decision five of six subjects misprice a card whose value is known to be two dollars.

\citet{bsk22} also estimate structural models of bidding. In every mechanism that displays a clock, including the static ones, behavior is best explained by a model in which the subject decides at each price whether to continue or stop. Only behavior in the sealed-bid second-price auction without a clock is best explained by a model that evaluates all bids at once.

This paper's explanation also predicts the timing of the passive clock's effect. In the passive clock treatment the clock runs after the bids are submitted, so a subject's first bid comes before she has seen a clock. From the second round on she has watched the clock resolve her own bid price by price. The comparison between exiting at one price and staying to the next is then available when she bids. The passive clock should therefore lag the dynamic clock by about one round and then catch up. Result 3 of \citet{bsk22} reports a lag of approximately one round when dropouts are not shown, and of two rounds when they are. They themselves trace the one-round lag to the clock being shown only after bids are submitted.

\subsection{Iterative deferred acceptance}
\label{app:ida}

Iterative deferred acceptance provides partial evidence of the kind discussed in Section~\ref{sec:dynamicmatching}. By \citet{ashlagigonczarowski18}, no stable matching mechanism has an obviously strategy-proof implementation for general priorities of the colleges. In iterative deferred acceptance, students apply to one college at a time and observe the result of each application before the next. \citet{bohakimov20} find that iterative deferred acceptance produces stable outcomes significantly more often than list submission. The reason is that more subjects follow the straightforward strategy, which applies at each round to the most preferred college among those still open to the student. The effect persists when cutoff information is withheld. I read this finding the way I read the ascending clock auction. Each round asks the student for one application. Suppose she compares applying to $a$ now with applying to $b$ now, and keeps the order of her later applications as it was. She then compares two lists that differ by the adjacent exchange of $a$ and $b$. Iterative deferred acceptance can show her, one round at a time, that the exchange changes her outcome only between $a$ and $b$.

The evidence is partial for two reasons. First, the straightforward strategy in the iterative mechanisms of \citet{bohakimov20} is an equilibrium strategy rather than a dominant strategy, so these mechanisms lie outside the class to which the prediction applies. Second, the iterative designs combine one choice at a time with feedback between rounds. Like the passive clock of \citet{bsk22}, they therefore do not separate salience from learning from feedback. The evidence is nonetheless consistent with this paper's explanation.

\subsection{Richness in Li's random serial dictatorship design}
\label{app:lidesign}

In the random serial dictatorship treatments of \citet{li17}, four agents share four objects, one unit each. Before she submits her list, each subject is told her own priority score, an integer from $1$ to $10$, but not the scores of the others. The others' scores are drawn independently and ties are broken at random, so every position remains possible whatever her score. The environment is therefore rich, and Corollary~\ref{cor:matching} applies to the design of \citet{li17}. Her score makes some positions unlikely, however. An agent who moves at position $k$ always receives one of the first $k$ objects on her list, so a pair at positions $j$ and $j+1$ is decisive for her only when $j \le k$. The agent who moves last receives the one remaining object whatever her list, and no pair is decisive for her. For a subject with a high score, the pairs low on her list are therefore decisive only in unlikely contingencies. The exchange of the two highest entries, by contrast, can be decisive whenever she does not move last. A list whose first two entries are the two objects she likes best, in the wrong order, therefore has a local improvement, namely the exchange that restores them. Exchanging the two objects she likes best is the most common error that \citet{li17} reports under serial dictatorship with list submission.

\section{Further discussion}
\label{app:discussion}

This appendix collects discussion that the main text summarizes. It covers the roles of the two requirements in Definitions~\ref{def:improves} and~\ref{def:lis}, the place of other conditions in Table~\ref{tab:four}, the ascending clock auction under one-situation neighbors, a proposed multi-account clock, and the pick-an-object serial dictatorship.

\subsection{The two requirements in the definition}
\label{app:remarks}

The requirement that a local improvement be strictly better in some contingency and the requirement that all locally unimprovable strategies be payoff-equivalent have different roles. The first requirement makes every run terminate, as Lemma~\ref{lem:terminate} shows. Without it, a run could move back and forth between two payoff-equivalent strategies. The second requirement makes the locally unimprovable strategies dominant. Termination alone does not. Consider a line on which moving one step to the right is a local improvement at every strategy below $a$. Suppose also that moving one step to the left is a local improvement at every strategy above $a + 1$. Every run ends at $a$ or at $a + 1$. If $a$ pays more in some contingencies and $a + 1$ in others, neither is dominant. The mechanism is then not local improvement solvable under these one-step neighbors, because its two locally unimprovable strategies are not payoff-equivalent. Payoff-equivalent locally unimprovable strategies, on the other hand, are allowed, and they occur in the applications. In the sealed-bid second-price auction with a value on the bid grid, the bid equal to the value and the bid one step above it give the same utility in every contingency. Both bids are locally unimprovable, and both are dominant. In an extensive-form mechanism, two strategies that specify the same action at every situation that both of them reach are payoff-equivalent. The definition therefore allows locally unimprovable strategies that differ only at situations they do not reach, without any further convention.

\subsection{Other conditions on a step}
\label{app:other}

Most other worst-best comparison concepts lie between the four conditions of Table~\ref{tab:four}. To place them, fix a situation $h$ that $s_{i}^{*}$ reaches and an action $a' \neq s^{*}_{h}$. As in Section~\ref{sec:extdef}, a continuation after an action at $h$ is the player's actions at the situations that follow the action. Dominance and obvious dominance compare every continuation after $a'$ with the continuation of $s_{i}^{*}$ after $s^{*}_{h}$, while action dominance and strong obvious dominance require the comparison for every continuation after $s^{*}_{h}$ as well. For action dominance this holds because, by perfect recall, no situation follows both actions, so the common strategy $r$ in (\ref{eq:AD}) does not link the two continuations. Two families of concepts fill the left column of Table~\ref{tab:four}. \citet{zhanglevin21} partition the contingencies and compare the worst case against the best case within each element of the partition. The finest partition gives dominance and the coarsest gives obvious dominance. \citet{arribillagamassoneme24} partition the players, hold fixed the strategies of the players in one's own group, and compare the worst case against the best case over the other players' strategies. The concepts of \citet{pyciatroyan23} fill the bottom row, from strong obvious dominance, where the player plans no future move, to obvious dominance, where she plans all of them.

The conditions of \citet{catoninixue} compare two actions at $h$ and treat the player's own continuation in a third way. Their baseline condition fixes a partition of $S(h)_{-i}$. The recommended action dominates $a'$ if every continuation after $a'$ is matched by some continuation after $s^{*}_{h}$. In each element of the partition, the worst payoff of the match must be at least the best payoff of the continuation it matches. With the singleton partition the baseline condition is their weak dominance between actions, and their obvious dominance between actions uses the trivial partition. Because their condition asks only that some match exist, each of the two concepts is implied at $h$ by both conditions in its row of Table~\ref{tab:four}. The continuation of $s_{i}^{*}$ supplies the match for dominance and obvious dominance, and any continuation supplies it for action dominance and strong obvious dominance.

When the two concepts are required at every situation that a strategy reaches, the freedom to choose the match adds nothing to dominance and obvious dominance. \citet{catoninixue} show that a strategy is dominant if and only if every action it prescribes at a situation it reaches is weakly dominant, and the same holds for obvious dominance. Since a finer partition gives a weaker condition, their family lies in the left column, between strategy-proofness and OSP.\footnote{The scenario-by-scenario dominance of \citet{catoninixue} lets the match change from one element of the partition to another, and their wishful dominance is the singleton case of scenario-by-scenario dominance. They also show that, required at every situation that a strategy reaches, these conditions are equivalent to the baseline ones and so stay in the left column.} Action-proofness implies weak dominance between actions at every situation that the recommended strategy reaches, and SOSP implies obvious dominance between actions there, but by Proposition~\ref{prop:taxonomy}(b) neither converse holds.

Their local dominance lies outside the table, since it treats two parts of $S(h)_{-i}$ differently. Where $s^{*}_{h}$ ends the player's play, the worst payoff of $s^{*}_{h}$ must be at least the best payoff of $a'$. Elsewhere every continuation after $a'$ must be repeated stage by stage after $s^{*}_{h}$ with the same outcome. By their results, local strategy-proofness implies strategy-proofness and coincides with OSP in static mechanisms. In dynamic games, they note, local dominance neither implies nor is implied by obvious or strong obvious dominance. Local strategy-proofness is also independent of action-proofness. The sealed-bid second-price auction is AP and not OSP, so it is not locally strategy-proof. The English auction with alternating bids, which they argue is locally strategy-proof, is not AP. Consider a type whose value lies above the price she would pay if she bid at her first turn and the rival then left. An action dominant strategy is dominant by Proposition~\ref{prop:taxonomy}(a), so it bids at that turn. Bidding there is worse than leaving at once if she keeps bidding above her value and the rival leaves only after the price passes it. No strategy of such a type is therefore action dominant.

\subsection{The ascending clock auction under one-situation neighbors}
\label{app:clockhamming}

Under one-situation neighbors the ascending clock auction is not local improvement solvable. As in the two-move game of Figure~\ref{fig:twomove}, some strategy can be improved only by changing two actions at once. A strategy specifies an action at every price at which the bidder could still be active, including prices she reaches only if she does not exit when she planned to. Take a value strictly between $0$ and $K - 1$ that is not a grid point. Consider a strategy that exits at a price $b$ below the recommended exit price $a^{*}$ but, if still in, would exit only at a price $b'$ above $a^{*}$. This strategy is locally unimprovable. Staying at $b$ alone moves the exit price to $b'$, which is not an improvement, since she then wins at a price above her value whenever the last rival exits at a price in $[a^{*}, b')$. Exiting earlier is never better, and a change at a price above $b$ alone has no effect, since the strategy has exited by then. An improvement must change two actions at once, staying at $b$ and exiting at a price above $b$ and no higher than $a^{*}$. Under one-step neighbors on exit prices, by contrast, the exit price $b$ has the neighbor $b + 1$, and Proposition~\ref{prop:spa} applies.

\subsection{A proposed design: the multi-account clock}
\label{sec:accounts}

The Ausubel auction asks a bidder at each price for a number of units. A demand stated at a price answers in one number the comparisons of Proposition~\ref{prop:vickrey} for her last unit and, if she gives it up, for the units before it. A design for future work would ask these comparisons one at a time, as the ascending clock auction does for a single unit.\footnote{This subsection is preliminary. It may be incorporated into a new version of the paper with the results of an experiment that tests this design, which I plan to run in the near future.}

The multi-account clock is defined as follows. Each bidder has one account for each unit she may receive, $L$ in all, ordered as her marginal values are, from the first unit to the last. All accounts are open when the auction starts. The price rises one grid step at a time, as in the Ausubel auction. At each price the bidder is asked, starting with her last open account, whether she still wants that account's unit at the current price. An answer of no closes the account and moves the question to the account before it, and an answer of yes ends the questions at that price. An account can therefore be closed only after every account after it has been closed, and the bidder's demand at each price is the number of her open accounts. Clinching, payments, and the end of the auction are as in the Ausubel auction, and a bidder observes only the current price and whether the auction is still open.

The multi-account clock and the Ausubel auction implement each other's outcomes. The answers at each price determine the demand there, and every sequence of demands of at most $L$ units that the Ausubel auction allows can be produced by such answers. Maps in both directions therefore send each strategy of either auction to a strategy of the other with the same demand at every price, and sincere bidding to sincere bidding. Sincere bidding, which keeps account $\ell$ open exactly while the price is below $\ta^{\ell}$, is weakly dominant in the multi-account clock, as it is in the Ausubel auction with no bid information. Since a bidder observes only the price and whether the auction is still open, a strategy in either auction amounts to a vector of exit prices, one for each unit. A unit she keeps at every price below $K$ has exit price $K$. The payoff of the strategy is that of the same vector of marginal bids in the static Vickrey auction. Under one-step neighbors on these vectors, both auctions are therefore local improvement solvable by Proposition~\ref{prop:vickrey}, for the same reason.

Neither auction is OSP for a bidder who demands two or more units. Reducing her demand by one unit, or closing an account other than the first, at a price below that unit's marginal value keeps her other units in play. The best outcome of that deviation, winning those units at a low price if the rivals' demand falls to zero, exceeds the worst outcome of the recommended action, under which she clinches nothing more. Obvious strategy-proofness on its own therefore predicts no difference between the two auctions.

The two auctions differ in the salience of the local comparison. Each question of the multi-account clock asks the bidder to compare keeping one unit for one more step with giving it up now, the rest of her strategy held fixed. The comparison involves one case, the contingencies in which the unit would clinch at the current price, and the unit's marginal value determines its sign. The Ausubel auction asks for the same comparisons through a single number, so a bidder who states a demand must separate them herself. Since demands that differ by two units are not neighbors, a question of the Ausubel auction is not confined to a choice between two neighboring strategies. The prediction for this design is therefore that the multi-account clock produces more dominant-strategy play than the Ausubel auction, which in turn produces more than the static Vickrey auction.

The multi-account clock is close to the way the Ausubel auction has been run in the laboratory.\footnote{In the Ausubel auction of \citet{engelmanngrimm09}, a bidder drops one unit by clicking one button and both units by clicking another. In the instructions used by \citet{kagellevin09}, each keystroke drops the bidder's lowest-valued remaining unit.} The laboratory interfaces let a bidder drop units with buttons or keystrokes instead of stating her demand as a number. With two units they differ from the multi-account clock mainly in not asking an explicit question for each open account. As far as I know, no experiment has asked a question for each open account or compared a unit-by-unit interface with one in which the bidder states her demand as a number. A test of the prediction with obvious strategy-proofness held fixed would compare the multi-account clock, the Ausubel auction with stated demands, and the static Vickrey auction. Every bidder would demand two or more units, and no one would see rivals' dropouts, which also excludes the information whose effect Section~\ref{sec:clock} attributes to learning from the environment.

\subsection{The pick-an-object serial dictatorship}
\label{app:pao}

For the pick-an-object serial dictatorship the recommended strategy is dominant. \citet{bohakimov23} state only that it is an equilibrium strategy there. They show, however, that every non-bossy rule with an OSP implementation, serial dictatorship among them, has a pick-an-object implementation in weakly dominant strategies. All agents first pick from the menu of all objects. The agents are then processed in the priority order. Each is assigned her current pick, and every agent with lower priority who picked the same object picks again from the objects that remain. The objects that remain when an agent is processed depend only on the agents with higher priority, who never observe her picks. Picking the most preferred object on each menu therefore gives her the best remaining object whatever the others do. The pick-an-object serial dictatorship therefore implements the outcomes of serial dictatorship with list submission.

\newpage
\bibliographystyle{aea}
\bibliography{ref_arxiv}

\end{document}